\documentclass[10pt]{article}

\usepackage{xspace}
\usepackage{ragged2e}
\usepackage{url}
\usepackage{mathtools}
\usepackage{amssymb}
\usepackage{amsthm}
\usepackage{empheq}
\usepackage{latexsym}
\usepackage{enumitem}
\usepackage{eurosym}
\usepackage{dsfont}
\usepackage{appendix}
\usepackage{color} 
\usepackage[unicode]{hyperref}
\usepackage{frcursive}
\usepackage[utf8]{inputenc}
\usepackage[T1]{fontenc}
\usepackage{geometry}
\usepackage{multirow}
\usepackage{todonotes}
\usepackage{lmodern}
\usepackage{subcaption}
\usepackage{anyfontsize}
\usepackage{stmaryrd}
\usepackage{natbib}
\usepackage{cleveref}
\usepackage[english]{babel}
\usepackage[english=british]{csquotes}
\usepackage{float}
\usepackage{bm}
\usepackage{natbib}
\everymath{\displaystyle}

\setcitestyle{numbers,open={[},close={]}}

\definecolor{red}{rgb}{0.7,0.15,0.15}
\definecolor{green}{rgb}{0,0.5,0}
\definecolor{blue}{rgb}{0,0,0.7}
\hypersetup{colorlinks, linkcolor={red},citecolor={green}, urlcolor={blue}}
			
\makeatletter \@addtoreset{equation}{section}

\newtheorem{theorem}{Theorem}[section]

\newtheorem{lemma}[theorem]{Lemma}

\newtheorem{definition}[theorem]{Definition}
\newtheorem{remark}[theorem]{Remark}

\def \E{\mathbb{E}}

\def \R{\mathbb{R}}

\def\Ac{{\cal A}}

\def\d{\mathrm{d}}

\title{A Mean Field Game between Informed Traders and a Broker}
\author{Philippe \textsc{Bergault}\footnote{Université Paris Dauphine-PSL, Ceremade, 75116, Paris, France; bergault@ceremade.dauphine.fr.} \and Leandro \textsc{Sánchez-Betancourt}\footnote{Mathematical Institute, and Oxford-Man Institute of Quantitative Finance, University of Oxford, OX2 6GG, Oxford, England; leandro.sanchezbetancourt@maths.ox.ac.uk.}}

\begin{document}

\maketitle

\begin{abstract}
We find closed-form solutions to the stochastic game between a broker and a mean-field of informed traders. In the finite player game,  the informed traders observe a common signal and a private signal. The broker, on the other hand, observes the trading speed of each of his clients and  provides liquidity to the informed traders. Each player in the game optimises wealth adjusted by inventory penalties. In the mean field version of the game, using a Gâteaux derivative approach, we characterise the solution to the game with a system of forward-backward stochastic differential equations that we solve explicitly. We find that the optimal trading strategy of the broker is linear on his own inventory,  on the average inventory among informed traders, and on the common signal or the average trading speed of the informed traders. The Nash equilibrium we find helps informed traders decide how to use private information, and helps brokers decide how much of the order flow they should externalise or internalise when facing a large number of clients. \\

\end{abstract}

\section{Introduction}

Liquidity provision plays a key role in financial markets. A large portion of the liquidity provision activity happens in over-the-counter (OTC) markets where broker-client relationships are ubiquitous. Brokers face an important trading problem when deciding how to face the order flow from informed traders. These problems are challenging because one is typically interested in the Nash equilibrium of a stochastic game.\\

The study of externalisation-internalisation strategies is an active area of research. Externalisation refers to the act of hedging or unwinding a position sent by a client. On the other hand, internalisation refers to the warehousing of risk by the broker, on the hope that either prices move favourably to the broker or that other trades arrive in the opposite direction. Focusing on electronic FX spot markets, \citeauthor{butz2019internalisation} \cite{butz2019internalisation} use queuing theory to derive a closed-form expression for the average internalisation horizon and the cost of internalisation. \citeauthor{barzykin2021algorithmic} \cite{barzykin2021algorithmic} propose a market making model for dealers who have access to an inter-dealer market allowing them to externalise part of their risk. In particular, they show that the dealer starts externalising only outside of a certain inventory range (see also \citeauthor{barzykin2022market} \cite{barzykin2022market, barzykin2023dealing}). The recent article of \citeauthor{cartea2023detecting} \cite{cartea2023detecting} uses a proprietary dataset of transactions of an FX broker to develop a framework that predicts toxic trades and uses this information to decide whether to internalise or externalise trades. Additionally, the recent BIS Triennial Survey in \citeauthor{schrimpf2019fx} \cite{schrimpf2019fx}, thoroughly discusses the trade-off between internalisation and externalisation on empirical grounds, highlighting the increasing prevalence of internalisation in FX markets. The paper shows diverse behaviours, ranging from complete externalisation to significant internalisation ratios. It is noteworthy that, despite internalisation ratios surpassing 80\% in the FX markets' top trading centres, hedging through externalisation remains a crucial aspect of risk management. \\

In a closely related branch of the literature, there are a number of works that study the unwinding of stochastic order flow. The work of \citeauthor{cartea2020trading} \cite{cartea2020trading} studies the optimal liquidation strategy of a broker trading in a triplet of currency pairs with stochastic order flow from their clients. 
In \citeauthor{cartea2022double} \cite{cartea2022double} the stochastic order flow to be unwound is that of the proceeds of the sale of a stock that trades in a foreign currency.
Recently, \citeauthor{muhle2023pre} \cite{muhle2023pre} studies how brokers pre-hedge a possible trade from a client to achieve (potentially) better outcomes for both parties.
Lastly, another recent article is that of \citeauthor{nutz2023unwinding} \cite{nutz2023unwinding} where the authors solve a control problem for the optimal externalisation schedule of an exogenous order flow with an Obizhaeva–Wang type price impact and quadratic instantaneous costs. 
Our paper also computes an optimal externalisation-internalisation strategy for the broker although it arises as the Nash equilibrium in a market with a large number of clients. \\

Information asymmetry has been studied extensively in the algorithmic trading literature. For instance, \citeauthor{muhlekarbe2017information} \cite{muhlekarbe2017information} show how short-term informational advantages can be monetised by high frequency traders. The competition between high frequency traders and slow traders with information asymmetry is also the topic of \citeauthor{cont2023fast} \cite{cont2023fast}. Liquidity provision with adverse selection is studied in \citeauthor{herdegen2023liquidity} \cite{herdegen2023liquidity}. Recently, \citeauthor{cartea2022brokers} \cite{cartea2022brokers} introduced a framework where a broker faces a representative informed trader and a representative uninformed trader. Their Stackelberg game admits closed-form solutions for the strategies of the informed trader and that of the broker. In this paper, we build on their framework and we design a problem where the broker faces a large number of informed traders, each of which have access to a common signal and a private signal. We then consider the mean-field-game (MFG) formulation of the problem and find closed-form solutions for the mean-field Nash equilibrium of the game. We show how the broker trades as a function of the average trading speed across informed traders, and how he manages inventory.\\

Our work lies at the intersection of algorithmic trading and mean-field games. Earlier works at this intersection were concerned with the standard optimal execution problem à la \citeauthor{almgren2001optimal} \cite{almgren2001optimal} studied from a mean-field game setting. For example,  in \citeauthor{cardaliaguet2018mean} \cite{cardaliaguet2018mean} the trader faces uncertainty with respect to price changes because of his actions but also has to deal with price changes due to other similar market participants. \citeauthor{huang2019mean} \cite{huang2019mean} extend this work using a major–minor mean-field game framework in which minor agents trade along with the major agent. In \citeauthor{firoozi2016mean} \cite{firoozi2016mean} the authors also consider an optimal execution problem through a linear-quadratic major-minor mean-field game, but the inventory of the major is only partially observed. The case of a large number of traders trying to perform optimal execution has been studied in \citeauthor{casgrain2018mean} \cite{casgrain2018mean, casgrain2020mean}. \citeauthor{neuman2023trading} \cite{neuman2023trading} study a similar problem with jointly aggregated transient price impact and a common price signal (see also \citeauthor{jaber2023equilibrium} \cite{jaber2023equilibrium}). Recently, the authors of \citeauthor{baldacci2023mean} \cite{baldacci2023mean} proposed a mean-field version of standard market making models à la \citeauthor{avellaneda2008high} \cite{avellaneda2008high} in which a market maker faces a large number of strategic market takers and sets his bid and ask quotes accordingly in order to manage inventory risk -- in particular, the broker cannot externalise. Our model departs from all these previous formulations while remaining closely related with respect to the end goal. To the best of our knowledge, this paper is the first to derive the unique closed-form solutions to the Nash equilibrium between a broker and a mean-field of informed traders. \\

This paper delves into the intricate strategy employed by a broker who not only provides liquidity to a large number of informed traders, but also engages in liquidity-taking transactions in a lit market. In our model, both the broker and the informed traders aim to maximise their expected wealth while strategically managing inventory holdings. The broker employs an inventory penalty to safeguard his strategy against inventory risks (especially toxic inventory). Simultaneously, the informed trader uses the inventory penalty to control her exposure to inventory risks stemming from speculative trades based on common and private signals. The problem is modelled as a linear-quadratic major-minor mean field game that we can solve explicitly using a Gâteaux derivative approach. The broker's strategy (i) determines the optimal externalisation of the flow from informed traders and (ii) guides the interactions with the lit market trading for hedging and speculative purposes. The derived closed-form strategy of the broker involves a linear combination of his inventory, the average informed trader’s inventory, and the common signal. For the representative informed trader, it also involves her private signal as well as her own inventory.\\

The remainder of the paper is organised as follows. Section \ref{sec_N_players} introduces a model with $N$ informed traders and the broker. Every agent observes a common signal on the price. On top of that, each informed trader observes a private signal, and the trading activity of the broker on the lit market has a permanent impact on the price. Section \ref{sec_mfg_proba_framework} derives the mean field limit of this game and solves for the Nash equilibrium. In particular, we show that the functionals we optimise are strict concave up to null sets, G\^ateaux differentiable, and we characterise the Nash equilibrium of the game with a system of forward-backward stochastic differential equations (FBSDEs) that we solve in closed form. Section \ref{sec_num} shows numerical results  and Section \ref{sec_concl} concludes.

\section{The game with $N$ informed traders}\label{sec_N_players}

\subsection{Framework}

We consider a trading horizon $T>0$ and a probability space $\big(\Omega,\mathcal F, \mathbb{P}\big)$ under which all the stochastic processes are defined. We set some positive integer $N\in \mathbb{N}^\star$, corresponding to the number of informed traders acting on a market consisting of a single asset whose price process is denoted by $S.$ This quantity can be thought as the mid-price of the asset.\\

We introduce $N+2$ independent standard Brownian motions $W^S, W^\alpha, W^1, \ldots, W^N$ that we employ in the equations below. Under the probability $\mathbb P$, the price process $(S_t)_{t\in [0,T]}$ follows
$$\d S_t = \sigma^S \d W^S_t,$$
where $\sigma^S>0$.\\

Each trader observes a common signal and a private signal. The common or fundamental signal $(\alpha_t)_{t\in [0,T]}$ satisfies the following stochastic differential equation (SDE)
$$\d \alpha_t = -k^\alpha \alpha_t \d t + \sigma^\alpha \d W^\alpha_t,$$
with $k^\alpha, \sigma^\alpha >0$.
On the other hand, the private signal of each trader is not shared, that is, the private signal of one trader is hidden to every other trader. For trader $n \in \llbracket 1,N \rrbracket$, we denote their signal by $(\alpha^n_t)_{t\in[0,T]}$, which follows the SDE
$$\d \alpha^n_t = -k^n \alpha^n_t \d t + \sigma^n \d W^n_t,$$
with $k^n, \sigma^n$ two positive random variables, independent from everything else, with the same joint distribution $\zeta$ for all $n$ and such that
$$\E [k^n] = \bar k \qquad \text{and} \qquad \E [\sigma^n] = \bar \sigma,$$
where $\bar k, \bar \sigma>0$.\\

We denote by $(Q^n_t)_{t\in [0,T]}$ and $(X^n_t)_{t\in [0,T]}$ the inventory process and cash process of the $n-$th trader, respectively. She trades with the broker at rate $(\nu^n_t)_{t\in [0,T]}$ (she buys when $\nu^n_t>0$ and sells when $\nu^n_t<0$). Therefore, her inventory has dynamics
$$\d Q^n_t = \nu^n_t \d t,$$
and her cash process has dynamics
$$\d X^n_t = -\nu^n_t \left(S_t + \eta^I \nu^n_t\right) \d t,$$
where $\eta^I>0$ is a transaction cost charged by the broker to the traders.\\

We denote by $(Q^B_t)_{t\in [0,T]}$ and $(X^B_t)_{t\in [0,T]}$ the inventory and cash process of the broker, respectively. The broker receives the order flow from the traders and trades in a lit market at rate $(N\nu^B_t)_{t\in [0,T]}$.\footnote{ The idea here is that the externalisation activity of the broker should scale with the number of participants. The control $\nu^B_t$ can be thought as the ‘‘execution rate per client’’. Of course, this is just a change of variable that does not change the problem for the $N-$player game.} Therefore, his inventory has dynamics
$$\d Q^B_t = \left( N\nu^B_t - \sum_{n=1}^N \nu^n_t \right) \d t,$$
and his cash process has dynamics
$$\d X^B_t = \sum_{n=1}^N \nu^n_t \left(S_t + \eta^I \nu^n_t\right) \d t - N\nu^B_t \left(S_t + \eta^B \nu^B_t \right)\d t,$$
where $\eta^B>0$ corresponds to the execution costs in the lit market.\\

For each $n\in \llbracket 1,N \rrbracket$, we introduce the (completed) filtration $\left( \mathcal F^n_t\right)_{t\in [0,T]}$ generated by $S, \alpha, \alpha^n, Q^n, \nu^B, k^n,\sigma^n$.
We also introduce the (completed) filtration $\left( \mathcal F^B_t\right)_{t\in [0,T]}$ generated by $S, \alpha, Q^B, \nu^1, \ldots, \nu^N$.


\subsection{The problem of the informed traders}

Let $b>0$. For $n\in \llbracket 1,N \rrbracket$, let us introduce the probability measure $\mathbb P^{n, \nu^B}$ given by
$$\left. \frac{\d \mathbb P^{n, \nu^B}}{ \d \mathbb P} \right|_{\mathcal F^n_t} = \exp \left( \int_0^t  \frac{b\,\nu^B_u + \alpha^n_u + \alpha_u}{\sigma^S}  \d W^S_u - \frac 12 \int_0^t \left(\frac{b\,\nu^B_u + \alpha^n_u + \alpha_u}{\sigma^S}\right)^2 \d u\right).$$

Under this probability measure, $(S_t)_{t\in [0,T]}$ has dynamics 
$$\d S_t = \left(b\,\nu^B_t + \alpha^n_t + \alpha_t \right) \d t + \sigma^S \d \tilde W^{S,n},$$
where $\tilde W^{S,n}$ is a standard Brownian motion under $\mathbb P^{n, \nu^B}$. In other words, each informed trader observes a common signal $\alpha$ on the price, on top of which each one observes an idiosyncratic signal $\alpha^n$ that is hidden to the other traders. Moreover, the broker has a linear permanent impact on the price due to the trading in the lit market.\\

For a given $(\nu^B_t)_{t\in[0,T]}$, the $n-$th informed trader maximises the following objective function
\begin{align*}
\mathbb E^{n, \nu^B} \left[  X^n_T + Q^n_T S_T - a^n\,\left( Q^n_T \right)^2 - \phi^n \int_0^T \left( Q^n_t \right)^2 \d t\right],
\end{align*}
over her set of admissible controls $(\nu^n_t)_{t\in [0,T]}$, where $\mathbb E^{n, \nu^B}$ is the expectation taken under probability $\mathbb P^{n, \nu^B}$, and $a^n,\phi^n$ two positive random variables, independent from everything else, with the same joint distribution $\xi$ for all $n$ and such that
$$\E [a^n] = \bar a \qquad \text{and} \qquad \E [\phi^n] = \bar \phi,$$
where $\bar a, \bar \phi >0$.\\

It is easy to see that this amounts to maximising
\begin{align}\label{sec_N_players:tradersproblem}
\mathbb E^{n, \nu^B} \left[  \int_0^T \left\{ Q^n_t \left(b\,\nu^B_t + \alpha^n_t + \alpha_t \right) - \eta^I \left( \nu^n_t\right)^2 - 2\,a^n Q^n_t \nu^n_t -  \phi^n\left( Q^n_t \right)^2 \right\} \d t\right].
\end{align}

\subsection{The problem of the broker} 

We introduce the probability measure $\mathbb P^{B, \nu^B}$ given by
$$\left. \frac{\d \mathbb P^{B, \nu^B}}{ \d\mathbb P} \right|_{\mathcal F^n_t} = \exp \left( \int_0^t  \frac{b\,\nu^B_u  + \alpha_u}{\sigma^S}  \d W^S_u - \frac 12 \int_0^t \left(\frac{b\,\nu^B_u + \alpha_u}{\sigma^S}\right)^2 \d u\right).$$

Under this probability measure, $(S_t)_{t\in [0,T]}$ has dynamics 
$$\d S_t = \left(b\,\nu^B_t  + \alpha_t \right) \d t + \sigma^S \d \tilde W^{S,B},$$
where $\tilde W^{S,B}$ is a standard Brownian motion under $\mathbb P^{B, \nu^B}$. That is, the broker observes the fundamental signal $\alpha$ on the price, and he has linear permanent impact on the price when he trades in the lit market.\\

For given $(\nu^1_t)_{t\in[0,T]}, \ldots, (\nu^N_t)_{t\in[0,T]}$, the broker wants to maximise the following objective function
\begin{align*}
\mathbb E^{B, \nu^B} \left[  X^B_T + Q^B_T S_T -\frac{a^B}{N}\,\left(Q^B_T\right)^2 - \frac{\phi^B}{N} \int_0^T \left( Q^B_t \right)^2 \d t\right],
\end{align*}
over his set of admissible controls $(\nu^B_t)_{t\in [0,T]}$, where $\mathbb E^{B, \nu^B}$ is the expectation taken under probability $\mathbb P^{B, \nu^B}$, and $a^B,\phi^B>0$ correspond to the risk aversion of the broker. 
It is easy to see that this amounts to maximising
\begin{align*}
\mathbb E^{B, \nu^B} \left[  \int_0^T \left\{ Q^B_t \left(b\,\nu^B_t  + \alpha_t \right) + \eta^I \sum_{n=1}^N \left(\nu^n_t \right)^2- N\eta^B \left( \nu^B_t\right)^2 - 2\frac{a^B}{N}Q^B_t \left( N\nu^B_t - \sum_{n=1}^N \nu^n_t \right) -  \frac{\phi^B}{N}\left( Q^B_t \right)^2 \right\} \d t\right].
\end{align*}

Of course, the optimisation problem remains unchanged if we scale the objective function by dividing it by $N$, in which case the broker maximises
\begin{align}\label{sec_N_players:brokersproblem}
\mathbb E^{B, \nu^B} \left[  \int_0^T \left\{ \bar Q^B_t \left(b\,\nu^B_t  + \alpha_t \right) + \eta^I \frac 1N \sum_{n=1}^N \left(\nu^n_t \right)^2- \eta^B \left( \nu^B_t\right)^2 - 2\,a^B \bar Q^B_t \left( \nu^B_t - \frac 1N \sum_{n=1}^N \nu^n_t \right) -  \phi^B\left( \bar Q^B_t \right)^2 \right\} \d t\right],
\end{align}
where $\left( \bar Q^B_t\right)_{t\in [0,T]} = \left( \frac{Q^B_t}{N}\right)_{t\in [0,T]}$, that is, 
$$\d \bar Q^B_t = \left( \nu^B_t - \frac 1N \sum_{n=1}^N \nu^n_t  \right) \d t.$$

\subsection{Limitations of the finite game model and mean field limit}\label{section_comments_limitations}

In the above, we described an agent-based model where $N$ informed traders trade against a single broker. Solving this multi-agent ($N+1$ players) problem boils down to the resolution of a system of Hamilton--Jacobi--Bellman equations, where the state variables are the inventory processes of the broker and the $N$ informed traders, as well as their idiosyncratic and common signals. This system of $N+1$ HJB equations is intractable in practice for a large number of informed traders.\\

To obtain an approximate solution to this problem, in the remaining of the paper we propose a mean field game approach. The broker does not face $N$ informed traders but infinitely many of them in a mean field interaction, which can be thought as the averaged behaviour of the informed traders. In the next section, we present rigorously the mean field limit of the $N$-players model, and the corresponding optimisation problems of the broker and a representative informed trader.\\ 

Propagation of chaos also tells us that at the limit, we should expect that the execution rate associated with each trader will become independent conditionally to $\alpha$ and $\nu^B$, therefore it is reasonable to assume that the dynamics of $\bar Q^B$ will converge toward
$$\d \bar Q^B_t = \left( \nu^B_t - \bar \nu_t  \right) \d t,$$
where $\bar \nu_t $ denotes the expectation of the execution rate of an informed trader at time $t$ knowing $\alpha_t$ and $\nu^B_t$.\\

In the next section, we rigorously introduce the mean field version of the problem.

\section{Facing many informed traders: the mean field limit}
\label{sec_mfg_proba_framework}
\subsection{Probabilistic framework}

\label{weakformsec}

We consider a trading horizon $T>0$. We denote by $\Omega_c$ the set of continuous functions from $[0,T]$ to $\mathbb R$, and let $\Omega = \R^4 \times \Omega_c^2$. The observable state is the canonical process $\chi = \left( k^I, \sigma^I, a^I, \phi^I, W^\alpha, W^I\right)$ of the space $\Omega$, with
$$k^I(\omega) = \omega^k, \quad \sigma^I(\omega) = \omega^\sigma, \quad a^I(\omega) = \omega^a, \quad \phi^I(\omega) = \omega^\phi, \quad W^\alpha_t(\omega)= \omega^\alpha(t), \quad \text{and} \quad W^I_t (\omega) = \omega^I(t)$$

for all $t\in[0,T]$ and $\omega = (\omega^k, \omega^\sigma, \omega^a, \omega^\phi, \omega^\alpha, \omega^I) \in \Omega$. We introduce positive constants $k^\alpha, \sigma^\alpha$, as well as the unique probability $\mathbb P$ on $\Omega$ under which the process $\left(W^\alpha, W^I\right)$ is a standard bi-dimensional Brownian motion, $(k^I, \sigma^I)$ has distribution $\zeta$, $(a^I, \phi^I)$ has distribution $\xi$, and  $\left(W^\alpha, W^I\right), (k^I, \sigma^I)$, and $(a^I, \phi^I)$ are independent.  \\

We define the processes $(\alpha_t)_{t\in [0,T]}$ and $(\alpha^I_t)_{t\in [0,T]}$ as the unique solution to the SDEs
$$\d \alpha_t = -k^\alpha \alpha_t \d t + \sigma^\alpha \d W^\alpha_t,$$
and
$$\d \alpha^I_t = -k^I \alpha^I_t \d t + \sigma^I \d W^I_t$$
respectively, with $\alpha_0 = \alpha^I_0 = 0$ and where $k^I, \sigma^I$ are two positive random variables with respective expectation $\bar k$ and $\bar \sigma $. \\

Finally, we denote the canonical $\mathbb P-$completed filtration generated by $\chi$ by $\mathbb F = \left( \mathcal F_t \right)_{t\in [0,T]}$. 

\subsection{Admissible controls}

The set of admissible strategies for the representative informed trader as well as the broker is given by
$$\mathcal A = \left\{ \nu = (\nu_t)_{t\in [0,T]} \left| \nu \text{ is } \mathbb F-\text{progressively measurable, and } \mathbb E \left[ \int_0^T \nu_t^2 \d t \right] < +\infty \right. \right\}.$$

Let us denote by $(\mu_t)_{t\in [0,T]}$ the process with values in $\mathcal P (\mathbb R)$ representing at time $t$ the distribution of the execution rates of the (other) informed traders conditionally to $\mathcal F_t$. We introduce the mean field execution rate $(\bar \nu_t)_{t\in [0,T]}$ given for each $t\in [0,T]$ by
$$\bar \nu_t = \int_{\mathbb R} x\, \mu_t(\d x).$$

For a couple $(\nu^I, \nu^B)\in \mathcal A^2$ of strategies, we define the associated inventory processes of the representative informed trader and the broker by 
$$Q^I_t = \int_0^t \nu^I_u \d u$$
and
$$\bar Q^B_t = \int_0^t \left( \nu^B_u - \bar \nu_u \right) \d u,$$
respectively.

\subsection{Optimisation problems and definition of equilibria}

Taking the mean field version of  \eqref{sec_N_players:tradersproblem}, we consider that the representative informed trader wants to solve
$$\underset{\nu^I \in \mathcal A}{\sup} H^{I,\nu^B} (\nu^I)$$
where
\begin{align}\label{representative_obj}
H^{I,\nu^B} (\nu^I) = \mathbb E \left[  \int_0^T \left\{ Q^I_t \left(b\,\nu^B_t + \alpha^I_t + \alpha_t \right) - \eta^I \left( \nu^I_t\right)^2 - 2\,a^I Q^I_t \nu^I_t -  \phi^I\left( Q^I_t \right)^2 \right\} \d t\right],
\end{align}
where $b>0$ represents the market impact of the broker, $\eta^I$ represents the transaction costs charged by the broker to the traders, and $a^I, \phi^I>0$ correspond to the risk aversion of the informed trader.\\

Next, taking the mean-field version of problem \eqref{sec_N_players:brokersproblem}, we consider the following problem for the broker
$$\underset{\nu^B \in \mathcal A}{\sup} H^{B,\mu}(\nu^B),$$
where
\begin{align}\label{broker_obj}
H^{B,\mu}(\nu^B) = \mathbb E \left[  \int_0^T \left\{ \bar Q^B_t \left(b\,\nu^B_t  + \alpha_t \right) + \eta^I \int_{\mathbb R} x^2 \mu_t(\d x) - \eta^B \left( \nu^B_t\right)^2 - 2\,a^B \bar Q^B_t \left(\nu^B_t - \int_\mathbb R x\, \mu_t(dx) \right) -  \phi^B\left( \bar Q^B_t \right)^2 \right\} \d t\right],
\end{align}
with $\eta^B>0$  the transaction costs on the lit market, and $a^B, \phi^B>0$ the risk aversion parameters for the broker. We assume that model parameters $a^B$ and $b$ satisfy that $2\,a^B - b \geq 0$, which will be necessary to prove the strict concavity (up to null sets) of the functional of the broker. We also assume that the permanent price impact parameter satisfies that  $b\leq 2\,\eta^B,\,2\,\eta^I,\,4\,\phi^B,\,4\,\bar{\phi}$ which is used to prove existence and uniqueness of a matrix Riccati differential equation below. \\

Finally, we can now define what we mean by a solution to the mean field game.
\begin{definition}\label{geneq}
A solution of the above game is given by a probability flow $\mu^{\star}\in \mathcal P(\mathbb R)$, a control $ \nu^{I,\star}\in \mathcal A$, and a control $\nu^{B,\star}\in \mathcal A$ such that
\begin{itemize}
    \item[$(i)$] $H^{I, \nu^{B,\star}}(\nu^{I,\star}) = \underset{\nu^I \in \mathcal A}{\sup} H^{I, \nu^{B,\star}} (\nu^I);$
    \item[$(ii)$] $H^{B, \mu^{\star}}(\nu^{B,\star}) = \underset{\nu^B \in \mathcal A}{\sup} H^{B, \mu^{\star}} (\nu^B);$
    \item[$(iii)$] $\mu^{\star}_t$ is the distribution of $\nu^{I,\star}_{t}$ conditionally to $\mathcal F^\alpha_t$ for Lebesgue--almost every $t\in[0,T]$,
\end{itemize}
where $\mathbb F^\alpha := \left( \mathcal F^\alpha_t\right)_{t\in [0,T]}$ is the $\mathbb P-$completed filtration generated by $W^\alpha$.
\end{definition}

\subsection{The informed trader's optimality condition}
In this subsection, we  show that (i) the functional $H^{I,\nu^B}$ is strictly concave, (ii) G\^ateaux differentiable, and (iii) we characterise the optimal trading strategy of the representative informed trader. 

\begin{lemma}\label{concaveI}
Let $\nu^B\in\mathcal A$.  The functional $H^{I,\nu^B}(\cdot):\mathcal A\to \R$ defined in \eqref{representative_obj} is strictly concave up to a $\mathbb P \otimes \d t-$null set.
\end{lemma}

\begin{proof}
Let $\nu^B\in\mathcal A$ and let $\zeta,\nu \in\mathcal A$. Let $A\in \mathcal A\otimes \mathcal B([0,T])$  with $\mathfrak{m}(A) > 0$ where $\mathfrak{m}:=\mathbb P \otimes \d t$ and  for $(\omega,t)\in A$ we have that $\eta_t(\omega) \neq \nu_t(\omega)$. Let $\rho\in(0,1)$, we need to show that 
\begin{equation*}
    H^{I,\nu^B}(\rho\,\zeta + (1-\rho)\,\nu) > \rho\,H^{I,\nu^B}(\zeta) + (1-\rho)\,H^{I,\nu^B}(\nu)\,.
\end{equation*}
We observe that 
\begin{equation*}
    Q^{I,\rho\,\zeta + (1-\rho)\,\nu}_t = \rho\,Q^{I,\zeta}_t + (1-\rho)\,Q^{I,\nu}_t\,,
\end{equation*}
where we use the controls in the superscript to draw attention to the process used to define $Q^I_t$. Then, it follows that 
\begin{align*}
     H^{I,\nu^B}(\rho\,\zeta + (1-\rho)\,\nu) &= \mathbb E \Bigg[  \int_0^T \bigg\{ Q^{I,\rho\,\zeta + (1-\rho)\,\nu}_t \left(b\,\nu^B_t + \alpha^I_t + \alpha_t \right) - \eta^I \left( \rho\,\zeta_t + (1-\rho)\,\nu_t\right)^2 \\
     &\qquad\qquad - 2\,a^I\,Q^{I,\rho\,\zeta + (1-\rho)\,\nu}_t \,\left(\rho\,\zeta_t + (1-\rho)\,\nu_t\right) -  \phi^I\left( Q^{I,\rho\,\zeta + (1-\rho)\,\nu}_t \right)^2 \bigg\} \d t\Bigg]\\
     &= \rho\,H^{I,\nu^B}(\zeta) +  (1-\rho)\,H^{I,\nu^B}(\nu) \\
     &\qquad 
     + \rho\,(1-\rho)\,\mathbb E \Bigg[  \int_0^T \bigg\{  -2\,\eta^I\,\zeta_t\,\nu_t  - 2\,a^I\left(Q^{I,\zeta}_t\,\nu_t + Q^{I,\nu}_t\,\zeta_t \right) - 2\,\phi^I\,Q^{I,\zeta}_t\,Q^{I,\nu}_t \bigg\} \d t\Bigg]\\
     &\qquad 
     - \rho\,(1-\rho)\,\mathbb E \Bigg[  \int_0^T \bigg\{ -\eta^I\,\zeta^2_t -\eta^I\,\nu^2_t  - 2\,a^I\left(Q^{I,\zeta}_t\,\zeta_t + Q^{I,\nu}_t\,\nu_t\right) \\
     &\hspace{6cm} -\phi^I\left(\left(Q^{I,\zeta}_t\right)^2 + \left(Q^{I,\nu}_t\right)^2\right)  \bigg\} \d t\Bigg]\,.
\end{align*}
It then follows that 
\begin{align*}
 H^{I,\nu^B}(\rho\,\zeta + (1-\rho)\,\nu) &=   \rho\,H^{I,\nu^B}(\zeta) +  (1-\rho)\,H^{I,\nu^B}(\nu) \\
 &\quad + \rho\,(1-\rho)\,\mathbb E \Bigg[  \int_0^T \bigg\{ \eta^I\left(\zeta_t - \nu_t\right)^2 + 2\,a^I\left(Q^{I,\zeta}_t - Q^{I,\nu}_t\right)\left(\zeta_t - \nu_t\right) + \phi^I\left(Q^{I,\zeta}_t + Q^{I,\nu}_t\right)^2 \bigg\} \d t\Bigg]\,.\nonumber
\end{align*}
By the definition of $Q^I$ we have that 
\begin{align*}
\mathbb{E}\left[\int_0^T \left( Q^{I,\zeta}_t - Q^{I,\nu}_t\right) \left( \zeta_t - \nu_t\right) \d t\right] & = \mathbb{E}\left[\int_0^T \left( \int_0^t \zeta_u - \nu_u\d u\right) \left( \zeta_t - \nu_t\right) \d t\right] \\
& = \frac{1}{2}\,\mathbb{E}\left[\int_0^T \int_0^T \left(\zeta_u - \nu_u\right) \left( \zeta_t - \nu_t\right) \d u\, \d t\right] \\
& = \frac{1}{2}\mathbb{E}\left[\left(\int_0^T \zeta_t - \nu_t\, \d t\right)^2\right]\,,
\end{align*}
and clearly
\begin{align*}
& \mathbb E \Bigg[  \int_0^T \left(Q^{I,\zeta}_t - Q^{I,\nu_t}\right)^2 \d t\Bigg] \geq  0 \,.
\end{align*}
Given that $\mathfrak{m}(A) > 0$, the following holds 
\begin{equation*}
    \mathbb E \Bigg[  \int_0^T \left(\zeta_t - \nu_t\right)^2 \d t\Bigg]  >0\,,
\end{equation*}
and this together with the above inequalities implies that $H^{I,\nu^B}(\rho\,\zeta + (1-\rho)\,\nu) > \rho\,H^{I,\nu^B}(\zeta) + (1-\rho)\,H^{I,\nu^B}(\nu)$.
\end{proof}

\begin{lemma}\label{gateauxI}
    The functional $H^{I,\nu^B}$ defined in \eqref{representative_obj} is everywhere Gâteaux differentiable in $\mathcal A$. The Gâteaux derivative at a point $\nu^I \in \mathcal A$ in a direction $w^I \in \mathcal A$ is given by
    \begin{align}\label{Gateaux_trader}
        \big\langle D H^{I,\nu^B} (\nu^I), w^I\big\rangle = \mathbb E \left[ \int_0^T w^I_t \left\{ - 2\,\eta^I \nu^I_t - 2\,a^I Q^I_T + \int_t^T \left( b\,\nu^B_u +\alpha^I_u + \alpha_u  - 2 \phi^I Q^I_u \right) \d u \right\} \d t   \right].
    \end{align}
\end{lemma}

\begin{proof}
    Let $\nu^I, w^I \in \Ac$. The Gâteaux derivative of $H^{I,\nu^B}$ at point $\nu^I$ in the direction of $w^I$, if it exists, is defined as
    $$\big\langle D H^{I,\nu^B} (\nu^I), w^I\big\rangle = \underset{\varepsilon \searrow 0}{\lim} \frac{H^{I,\nu^B} (\nu^I+\varepsilon w^I) - H^{I,\nu^B} (\nu^I)}{\varepsilon}. $$
    Let $\varepsilon >0$, and define 
    $$\tilde Q^I_t = \int_0^t w^I_t \d t,$$
    for all $t\in [0,T]$. Then a direct computation gives us
    \begin{align*}
        H^{I,\nu^B} (\nu^I+\varepsilon w^I) &= H^{I,\nu^B} (\nu^I) + \varepsilon \mathbb E \Bigg[\int_0^T \bigg\{ \tilde Q^I_t \Big(b\,\nu^B_t + \alpha^I_t +\alpha_t \Big) - 2\,\eta^I \nu^I_t w^I_t - 2\,a^I Q^I_t w^I_t - 2\,a^I \tilde Q^I_t \nu^I_t  - 2\,\phi^I Q^I_t \tilde Q^I_t \bigg\} \d t \Bigg]\\
        &\quad + \varepsilon^2 \mathbb E \left[ \int_0^T \left\{ -\eta^I \left( w^I_t \right)^2 - 2\,a^I \tilde Q^I_t w^I_t - \phi^I \left( \tilde Q^I_t \right)^2 \right\}\d t \right].
    \end{align*}
    Let us denote by $A$ the term
    $$A:= \mathbb E \left[ \int_0^T \left\{ -\eta^I \left( w^I_t \right)^2 - 2\,a^I \tilde Q^I_t w^I_t - \phi^I \left( \tilde Q^I_t \right)^2 \right\}\d t \right].$$
    Therefore, we have
    \begin{align*}
        \frac{H^{I,\nu^B} \!(\nu^I+\varepsilon w^I)\! -\! H^{I,\nu^B}\! (\nu^I)}{\varepsilon} &=  \mathbb E\! \Bigg[\!\int_0^T \!\!\bigg\{\! \tilde Q^I_t \Big(b\,\nu^B_t + \alpha^I_t +\alpha_t \Big) - 2\,\eta^I \nu^I_t w^I_t - 2\,a^I Q^I_t w^I_t - 2\,a^I \tilde Q^I_t \nu^I_t  - 2\,\phi^I Q^I_t \tilde Q^I_t\! \bigg\} \d t \Bigg]\! +\! \varepsilon A.
    \end{align*}
    We can write this as
    \begin{align*}
        \frac{H^{I,\nu^B} \!(\nu^I+\varepsilon w^I)\! -\! H^{I,\nu^B}\! (\nu^I)}{\varepsilon} &=  \mathbb E\! \Bigg[\!\int_0^T \!\! w^I_t \bigg\{\! - 2\,\eta^I \nu^I_t - 2\,a^I Q^I_T + \int_t^T \Big(b\,\nu^B_u + \alpha^I_u + \alpha_u  - 2\,\phi^I Q^I_u \Big) \d u  \! \bigg\} \d t \Bigg]\! +\! \varepsilon A.
    \end{align*}
    Taking the limit as $\varepsilon \searrow 0$ finally yields the result.
    
\end{proof}

\begin{theorem}
    We have that 
    $$\nu^{I,\star} = \underset{\nu^I \in \mathcal A}{\arg \max }\  H^{I,\nu^B} (\nu^I)$$
    if and only if $\nu^{I,\star}$ is the unique strong solution to the FBSDE
    \begin{align}\label{BSDE_trader}
    \begin{cases}
        -\d \left(2\,\eta^I \nu^{I,\star}_t \right) &= \left(b\,\nu^B_t + \alpha^I_t + \alpha_t - 2 \phi^I Q^{I,\star}_t \right) \d t - \d Z^I_t,\\
        2\,\eta^I \nu^{I,\star}_T &=-2\,a^I Q^{I, \star}_T,
    \end{cases}
    \end{align}
    where $Z^I \in \mathbb H^2_T$ is an $\mathbb F-$adapted $\mathbb P-$martingale.
\end{theorem}

\begin{proof}
    As in \citeauthor{casgrain2020mean} \cite{casgrain2020mean}, using Lemmas \ref{concaveI} and \ref{gateauxI}, we can apply the result of \citeauthor{ekeland1999convex} \cite{ekeland1999convex} which states that
    $$ \big\langle D H^{I,\nu^B} (\nu^{I,\star}), w^I\big\rangle =0\quad \forall w^I \in \Ac \iff \nu^{I,\star} = \underset{\nu^I \in \mathcal A}{\arg \max }\  H^{I,\nu^B} (\nu^I).$$
    Therefore, it only remains to prove that $\big\langle D H^{I,\nu^B} (\nu^{I,\star}), w^I\big\rangle =0\quad \forall w^I \in \Ac$ if and only if $\nu^{I,\star}$ is solution to the FBSDE \eqref{BSDE_trader}.\\

    Let us first assume that $\big\langle D H^{I,\nu^B} (\nu^{I,\star}), w^I\big\rangle =0$ for all $w^I \in \Ac$. This implies that 
    $$\mathbb E \left[ - 2\,\eta^I \nu^{I,\star}_t - 2\,a^I Q^I_T + \int_t^T \left( b\,\nu^B_u +\alpha^I_u + \alpha_u - 2 \phi^I Q^{I,\star}_u \right) \d u \bigg| \mathcal F_t  \right] = 0$$
    almost surely for all $t\in [0,T]$. Therefore,
    \begin{align*}
        -2\,\eta^I &\nu^{I,\star}_t =  \mathbb E \left[ 2\,a^I Q^{I,\star}_T -  \int_t^T \left( b\,\nu^B_u +\alpha^I_u + \alpha_u  - 2 \phi^I Q^{I,\star}_u \right) \d u \bigg| \mathcal F_t  \right]\\
        &= \int_0^t\!\!\! \left( b\,\nu^B_u +\alpha^I_u + \alpha_u  - 2 \phi^I Q^{I,\star}_u \right)\! \d u + \mathbb E \left[ \! 2\,a^I Q^{I,\star}_T \!-\!  \int_0^T\!\!\! \left( b\,\nu^B_u +\alpha^I_u + \alpha_u - 2 \phi^I Q^{I,\star}_u \right)\! \d u \bigg| \mathcal F_t  \right]\\
        &= \int_0^t \left( b\,\nu^B_u +\alpha^I_u + \alpha_u - 2 \phi^I Q^{I,\star}_u \right) \d u - Z^I_t,
    \end{align*}
    where the process $Z^I$ given by
    $$Z^I_t := -\mathbb E \left[ 2\,a^I Q^{I,\star}_T \!-\!  \int_0^T \left( b\,\nu^B_u +\alpha^I_u + \alpha_u  - 2 \phi^I Q^{I,\star}_u \right) \d u \bigg| \mathcal F_t  \right]$$
    is a martingale, by definition. Hence it is clear that $\nu^{I,\star}$ is solution to the FBSDE \eqref{BSDE_trader}.\\

    Conversely, assume that $\nu^{I,\star}$ is solution to the FBSDE \eqref{BSDE_trader}. Then $\nu^{I,\star}$ can be represented implicitly as
    \begin{align*}
        2\,\eta^I \nu^{I,\star}_t = \mathbb E \Bigg[ - 2\,a^I Q^{I,\star}_T + \int_t^T\left( b\,\nu^B_u +\alpha^I_u + \alpha_u - 2 \phi^I Q^{I,\star}_u \right) \d u  \bigg| \mathcal F_t\Bigg].
    \end{align*}
    Plugging this into the expression of the Gâteaux derivative \eqref{Gateaux_trader} in Lemma \ref{gateauxI}, it is clear that it vanishes almost surely for any $w^I \in \Ac$.
\end{proof}

\subsection{The broker's optimality condition}
Similar to the above, in this subsection, we  show that (i) the functional $H^{B,\mu}$ is strictly concave, (ii) G\^ateaux differentiable, and (iii) we characterise the optimal trading strategy of the broker.

\begin{lemma}\label{concaveB}
Let  $(\mu_t)_{t\in [0,T]}$ with values in $\mathcal P (\mathbb R)$ be the distribution of the execution rates of the informed traders conditionally to $\mathcal F_t$.    The functional $H^{B,\mu}(\cdot):\mathcal A\to \R$ defined in \eqref{broker_obj} is strictly concave up to a $\mathbb P \otimes \d t-$null set.
\end{lemma}

\begin{proof}
Fix $(\mu_t)_{t\in [0,T]}$ and let $\zeta,\nu \in\mathcal A$. Let $A\in \mathcal A\otimes \mathcal B([0,T])$  with $\mathfrak{m}(A) > 0$ where $\mathfrak{m}:=\mathbb P \otimes \d t$ and  for $(\omega,t)\in A$ we have that $\eta_t(\omega) \neq \nu_t(\omega)$. Let $\rho\in(0,1)$, we need to show that 
\begin{equation*}
    H^{B,\mu}(\rho\,\zeta + (1-\rho)\,\nu) > \rho\,H^{B,\mu}(\zeta) + (1-\rho)\,H^{B,\mu}(\nu)\,.
\end{equation*}
Similar to the proof of strict concavity for the informed trader, observe that 
\begin{equation*}
    \bar{Q}^{B,\rho\,\zeta + (1-\rho)\,\nu}_t = \int_0^t \left( \rho\,\zeta_u + (1-\rho)\,\nu_u - \bar \nu_u \right) \d u = \rho\,\bar{Q}^{B,\zeta}_t + (1-\rho)\,\bar{Q}^{B,\nu}_t\,,
\end{equation*}
where (as before) we use the notation of having the controls in the superscript.
It follows that 
\begin{align*}
    H^{B,\mu}(\rho\,\zeta + (1-\rho)\,\nu) &= \mathbb E \Bigg[  \int_0^T \bigg\{ \bar Q^{B,\rho\,\zeta + (1-\rho)\,\nu}_t \left(b\left(\rho\,\zeta_t + (1-\rho)\,\nu_t\right)  + \alpha_t \right) + \eta^I \int_{\mathbb R} x^2 \mu_t(\d x) - \eta^B \left( \rho\,\zeta_t + (1-\rho)\,\nu_t\right)^2 \\
    &\qquad\qquad\quad - 2\,a^B \bar Q^{B,\rho\,\zeta + (1-\rho)\,\nu}_t \left(\rho\,\zeta_t + (1-\rho)\,\nu_t - \int_\mathbb R x\, \mu_t(dx) \right) -  \phi^B\left( \bar Q^{B,\rho\,\zeta + (1-\rho)\,\nu}_t \right)^2 \bigg\} \d t\Bigg]\\
    & = \rho\,H^{B,\mu}(\zeta) + (1-\rho)\,H^{B,\mu}(\nu)\\
    & \qquad 
    + \rho\,(1-\rho)\,\mathbb E \Bigg[  \int_0^T \bigg\{
    + \bar{Q}^{B,\zeta}_t\, b\,\nu_t + \bar{Q}^{B,\nu}_t\, b\,\zeta_t  - 2\,\eta^B\,\zeta_t\,\nu_t \\
    & \hspace{4cm} - 2\,a^B\,\left(\bar{Q}^{B,\zeta}_t\,\nu_t +\bar{Q}^{B,\nu}_t\,\zeta_t\right) - 2\,\phi^B\,\bar{Q}^{B,\zeta}_t\,\bar{Q}^{B,\nu}_t
    \bigg\} \d t\Bigg]\\
    & \qquad 
    - \rho\,(1-\rho)\,\mathbb E \Bigg[  \int_0^T \bigg\{   
    + \bar{Q}^{B,\zeta}_t\, b\,\zeta_t + \bar{Q}^{B,\nu}_t\, b\,\nu_t - \eta^B\,\left(\zeta^2_t -\nu^2_t\right) \\
    & \hspace{4cm} - 2\,a^B\,\left(\bar{Q}^{B,\zeta}_t\,\zeta_t + \bar{Q}^{B,\nu}_t\,\nu_t\right) - \phi^B\,\left(\bar{Q}^{B,\zeta}_t\right)^2 - \phi^B\,\left(\bar{Q}^{B,\nu}_t\right)^2
    \bigg\} \d t\Bigg]\,.
\end{align*}
It then follows that 
\begin{align*}
    & H^{B,\mu}(\rho\,\zeta + (1-\rho)\,\nu) = \rho\,H^{B,\mu}(\zeta) + (1-\rho)\,H^{B,\mu}(\nu)\\
    & \qquad 
    + \rho\,(1-\rho)\,\mathbb E \Bigg[  \int_0^T \bigg\{
    \left(2\,a^B - b\right)\,\left(\bar{Q}^{B,\zeta}_t - \bar{Q}^{B,\nu}_t\right)\left(\zeta_t -\nu_t\right) 
    +\eta^B\left(\zeta_t - \nu_t\right)^2 +\phi^B\left(\bar{Q}^{B,\zeta}_t - \bar{Q}^{B,\nu}_t\right)^2
    \bigg\} \d t\Bigg]\,.
\end{align*}
Similar to above, 
\begin{align*}
\mathbb{E}\left[\int_0^T \left( \bar{Q}^{B,\zeta}_t - \bar{Q}^{B,\nu}_t\right) \left( \zeta_t - \nu_t\right) \d t\right] & =  \frac{1}{2}\mathbb{E}\left[\left(\int_0^T \zeta_t - \nu_t\, \d t\right)^2\right]\geq 0\,,
\end{align*}
and clearly
\begin{equation*}
    \mathbb{E}\left[\int_0^T \left(\bar{Q}^{B,\zeta}_t - \bar{Q}^{B,\nu}_t\right)^2 \d t\right] \geq 0\,.
\end{equation*}
Given that $\mathfrak{m}(A)>0$, then 
\begin{equation*}
    \mathbb{E}\left[\int_0^T \left(\zeta_t - \nu_t\right)^2  \d t\right] > 0\,,
\end{equation*}
and this implies that $H^{B,\mu}(\rho\,\zeta + (1-\rho)\,\nu) >
\rho\,H^{B,\mu}(\zeta) + (1-\rho)\,H^{B,\mu}(\nu)$.
\end{proof}

\begin{lemma}\label{gateauxB}
    The functional $H^{B,\mu}$ defined in \eqref{broker_obj} is everywhere Gâteaux differentiable in $\mathcal A$. The Gâteaux derivative at a point $\nu^B \in \mathcal A$ in a direction $w^b \in \mathcal A$ is given by
    \begin{align}\label{Gateaux_broker}
        \big\langle D H^{B,\mu} (\nu^B), w^B\big\rangle = \mathbb E\! \left[ \int_0^T\!\!\!\!\! w^B_t \left\{ (b-2\,a^B) \bar Q^B_T \!-\! 2\,\eta^B \nu^B_t +\! \int_t^T\!\!\!\! \left( b\,\int_\mathbb R x\, \mu_u(\d x) + \alpha_u  - 2 \phi^B \bar Q^B_u \right)\! \d u \right\}\! \d t   \right]\!.
    \end{align}
\end{lemma}

\begin{proof}
    Let $\nu^B, w^B \in \Ac$. The Gâteaux derivative of $H^{B,\mu}$ at point $\nu^B$ in the direction of $w^B$, if it exists, is defined as
    $$\big\langle D H^{B,\mu} (\nu^B), w^B\big\rangle = \underset{\varepsilon \searrow 0}{\lim} \frac{H^{B,\mu} (\nu^B+\varepsilon w^B) - H^{B,\mu} (\nu^B)}{\varepsilon}. $$
    Let $\varepsilon >0$, and define 
    $$\tilde Q^B_t = \int_0^t w^B_t \d t,$$
    for all $t\in [0,T]$. Then a direct computation gives us
    \begin{align*}
        H^{B,\mu} (\nu^B+\varepsilon w^B) &= H^{B,\mu} (\nu^B) + \varepsilon \mathbb E \Bigg[\int_0^T \bigg\{ \tilde Q^B_t \Big(b\,\nu^B_t +\alpha_t \Big) + \Big(b-2\,a^B\Big)\bar Q^B_t w^B_t - 2\,\eta^B \nu^B_t w^B_t- 2\,a^B \bar Q^B_t \tilde Q^B_t\\
        &\quad- 2\,a^B \tilde Q^B_t\Big( \nu^B_t - \int_\mathbb R x\, \mu_t(\d x)\Big)  - 2\,\phi^B \bar Q^B_t \tilde Q^B_t \bigg\} \d t \Bigg]\\
        &\quad + \varepsilon^2 \mathbb E \left[ \int_0^T \left\{ \Big(b-2\,a^B\Big)\tilde Q^B_t w^B_t - \eta^B \Big( w^B_t\Big)^2 - \phi^B \Big( \tilde Q^B_t\Big)^2 \right\}\d t \right].
    \end{align*}
    Let us denote by $A$ the term
    $$A:= \mathbb E \left[ \int_0^T \left\{ \Big(b-2\,a^B\Big)\tilde Q^B_t w^B_t - \eta^B \Big( w^B_t\Big)^2 - \phi^B \Big( \tilde Q^B_t\Big)^2 \right\}\d t \right].$$
    Therefore, we have
    \begin{align*}
        \frac{H^{B,\mu} \!(\nu^B+\varepsilon w^B)\! -\! H^{B,\mu}\! (\nu^B)}{\varepsilon} &=  \mathbb E \Bigg[\int_0^T \bigg\{ \tilde Q^B_t \Big(b\,\nu^B_t +\alpha_t \Big) + \Big(b-2\,a^B\Big)\bar Q^B_t w^B_t - 2\,\eta^B \nu^B_t w^B_t- 2\,a^B \bar Q^B_t \tilde Q^B_t\\
        &\quad- 2\,a^B \tilde Q^B_t\Big( \nu^B_t - \int_\mathbb R x\, \mu_t(\d x)\Big)  - 2\,\phi^B \bar Q^B_t \tilde Q^B_t \bigg\} \d t \Bigg]\! +\! \varepsilon A.
    \end{align*}
    We can write this as
    \begin{align*}
         \frac{H^{B,\mu} \!(\nu^B+\varepsilon w^B)\! -\! H^{B,\mu}\! (\nu^B)}{\varepsilon} &=\mathbb E \Bigg[ \int_0^T w^B_t \bigg\{ \Big(b-2\,a^B \Big)\bar Q^B_T - 2\,\eta^B \nu^B_t\\
         &\quad + \int_t^T \Big(b\int_\mathbb R x\, \mu_u(\d x) + \alpha_u -2\,\phi^B \bar Q^B_u  \Big) \d u\bigg\} \d t  \Bigg]   +\! \varepsilon A.
    \end{align*}
    Taking the limit as $\varepsilon \searrow 0$ finally yields the result.
\end{proof}

\begin{theorem}
    We have that 
    $$\nu^{B,\star} = \underset{\nu^B \in \mathcal A}{\arg \max }\  H^{B,\mu} (\nu^B)$$
    if and only if $\nu^{B,\star}$ is the unique strong solution to the FBSDE
    \begin{align}\label{BSDE_broker}
    \begin{cases}
        -\d \left(2\,\eta^B \nu^{B,\star}_t \right) &= \left(b\,\bar \nu_t + \alpha_t  - 2 \phi^B \bar Q^{B,\star}_t \right) \d t - \d Z^B_t,\\
        2\,\eta^B \nu^{B,\star}_T &=(b-2\,a^B)\bar Q^{B,\star}_T,
    \end{cases}
    \end{align}
    where $Z^B \in \mathbb H^2_T$ is an $\mathbb F-$adapted $\mathbb P-$martingale.
\end{theorem}

\begin{proof}
    As before, using Lemmas \ref{concaveB} and \ref{gateauxB}, we can apply the result of \citeauthor{ekeland1999convex} \cite{ekeland1999convex} which states that
    $$ \big\langle D H^{B,\mu} (\nu^{B,\star}), w^B\big\rangle =0\quad \forall w^B \in \Ac \iff \nu^{B,\star} = \underset{\nu^B \in \mathcal A}{\arg \max }\  H^{B,\mu} (\nu^B).$$
    Therefore, it only remains to prove that $\big\langle D H^{B,\mu} (\nu^{B,\star}), w^B\big\rangle =0\quad \forall w^B \in \Ac$ if and only if $\nu^{B,\star}$ is solution to the FBSDE in \eqref{BSDE_broker}.\\

    Let us first assume that $\big\langle D H^{B,\mu} (\nu^{B,\star}), w^B\big\rangle =0\quad \forall w^B \in \Ac$. This implies that 
    $$\mathbb E \left[ \Big(b-2\,a^B \Big)\bar Q^{B,\star}_T - 2\,\eta^B \nu^{B,\star}_t + \int_t^T \Big(b\int_\mathbb R x\, \mu_u(\d x) + \alpha_u -2\,\phi^B \bar Q^{B,\star}_u  \Big) \d u \bigg| \mathcal F_t  \right] = 0$$
    almost surely for all $t\in [0,T]$. Therefore,
    \begin{align*}
        -2\,\eta^B \nu^{B,\star}_t &=  \mathbb E \left[- \Big(b-2\,a^B \Big)\bar Q^{B,\star}_T - \int_t^T \Big(b\!\!\int_\mathbb R x\, \mu_u(\d x) + \alpha_u -2\,\phi^B \bar Q^{B,\star}_u  \Big) \d u \bigg| \mathcal F_t  \right]\\
        &= \int_0^t\!\!\!  \Big(b\int_\mathbb R x\, \mu_u(\d x) + \alpha_u -2\,\phi^B \bar Q^{B,\star}_u  \Big)\! \d u - \mathbb E \left[ \!  \Big(b-2\,a^B \Big)\bar Q^{B,\star}_T + \int_0^T\!\!\! \Big(b\!\!\int_\mathbb R x\, \mu_u(\d x) + \alpha_u -2\,\phi^B \bar Q^{B,\star}_u  \Big) \d u \bigg| \mathcal F_t  \right]\\
        &= \int_0^t\!\!\!  \Big(b\int_\mathbb R x\, \mu_u(\d x) + \alpha_u -2\,\phi^B \bar Q^{B,\star}_u  \Big)\! \d u - Z^B_t,
    \end{align*}
    where the process $Z^B$ given by
    $$Z^B_t := \mathbb E \left[ \!  \Big(b-2\,a^B \Big)\bar Q^{B,\star}_T + \int_0^T\!\!\! \Big(b\!\!\int_\mathbb R x\, \mu_u(\d x) + \alpha_u -2\,\phi^B \bar Q^{B,\star}_u  \Big) \d u \bigg| \mathcal F_t  \right]$$
    is a martingale, by definition. Hence it is clear that $\nu^{B,\star}$ is solution to the FBSDE \eqref{BSDE_broker}.\\

    Conversely, assume that $\nu^{B,\star}$ is solution to the FBSDE \eqref{BSDE_broker}. Then $\nu^{B,\star}$ can be represented implicitly as
    \begin{align*}
        2\,\eta^B \nu^{B,\star}_t = \mathbb E \left[ \!  \Big(b-2\,a^B \Big)\bar Q^{B,\star}_T + \int_t^T\!\!\! \Big(b\!\!\int_\mathbb R x\, \mu_u(\d x) + \alpha_u -2\,\phi^B \bar Q^{B,\star}_u  \Big) \d u \bigg| \mathcal F_t  \right].
    \end{align*}
    Plugging this into the expression of the Gâteaux derivative \eqref{Gateaux_broker} in Lemma \ref{gateauxB}, it is clear that it vanishes almost surely for any $w^B \in \Ac$.
\end{proof}

\subsection{The mean field FBSDE system}

Finally, at  equilibrium, we have the following system of coupled FBSDEs
\begin{align}\label{MFGFBSDE_system}
\begin{cases}
    -\d \left(2\,\eta^I \nu^{I,\star}_t \right) &= \left(b\,\nu^B_t + \alpha^I_t + \alpha_t - 2\,\phi^I Q^{I,\star}_t \right) \d t - \d Z^I_t,\\
    -\d \left(2\,\eta^B \nu^{B,\star}_t \right) &= \left(b\,\bar \nu^{\star}_t + \alpha_t  - 2 \phi^B \bar Q^{B,\star}_t \right) \d t - \d Z^B_t,\\
        2\,\eta^I \nu^{I,\star}_T &= -2\,a^I Q^{I, \star}_T\\
        2\,\eta^B \nu^{B,\star}_T &= - (2\,a^B-b)\bar Q^{B,\star}_T.
\end{cases}
\end{align}
Moreover, we know that at the equilibrium
\begin{align}\label{cond_nubar}
    \bar \nu^{\star}_t = \mathbb E \left[ \nu^{I,\star} \left| \mathcal F^\alpha_t \right. \right].
\end{align}

\subsection{The closed-form solution to the FBSDE}
Next, we solve the FBSDE systems we derived. First, we solve the mean-field system of FBSDEs explicitly, and then we solve the FBSDE of an individual informed trader.

\subsubsection{The optimal strategy of the broker}

In the equilibrium, we solve the FBSDE system
\begin{align}\label{MFGFBSDE_system2}
\begin{cases}
    -\d \left(2\,\eta^I \bar{\nu}^{\star}_t \right) &= \left(b\,\nu^{B,\star}_t  + \alpha_t - 2\,\bar \phi \bar{Q}^{\star}_t \right) \d t - \d \bar{Z}^I_t,\\
    -\d \left(2\,\eta^B \nu^{B,\star}_t \right) &= \left(b\,\bar \nu^{\star}_t + \alpha_t  - 2 \phi^B \bar Q^{B,\star}_t \right) \d t - \d Z^B_t,\\
        2\,\eta^I \bar{\nu}^{\star}_T &= -2\,\bar a \bar{Q}^{\star}_T\\
        2\,\eta^B \nu^{B,\star}_T &= - (2\,a^B-b)\bar Q^{B,\star}_T.
\end{cases}
\end{align}

Let us make an ansatz, and look for the solution to the above system in the form
\begin{align*}
    \bar{\nu}^{\star}_t &= g^a_t \alpha_t + g^b_t\,\bar{Q}^{\star}_t + g^c_t\,\bar{Q}^{B,\star}_t\,,\\
    \nu^{B,\star}_t &= h^a_t \alpha_t + h^b_t\,\bar{Q}^{\star}_t + h^c_t\,\bar{Q}^{B,\star}_t\,,
\end{align*}
where $g^a_t,g^b_t,g^c_t$ and $h^a_t,h^b_t,h^c_t$ are deterministic $\mathcal C^1$ functions, with terminal conditions $g^a_T = h^a_T = g^c_T = h^b_T = 0$, $g^b_T = -\bar a /\eta^I$ and  $h^c_T = -(2\,a^B-b) /2\,\eta^B$, and where 
\begin{equation*}
    \bar{Q}^{\star}_t =  \int_0^t \bar{\nu}^\star_u \,\d u\,, \qquad \text{ and }\qquad \bar{Q}^{B,\star}_t =  \int_0^t \left( \nu^{B,\star}_u - \bar{\nu}^\star_u \right)\,\d u \,.
\end{equation*}

It then follows that  
\begin{align*}
    \d \bar{\nu}^{\star}_t &=  \alpha_t \d g^a_t + g^a_t \d \alpha_t + \bar{Q}^{\star}_t \,\d g^b_t  + g^b_t \,\d \bar{Q}^{\star}_t + \bar{Q}^{B,\star}_t\, \d g^c_t + g^c_t\,\d \bar{Q}^{B,\star}_t\\
&=  \alpha_t \d g^a_t - k^\alpha \alpha_t g^a_t \d t  + \bar{Q}^{\star}_t \,\d g^b_t  + g^b_t \,\bar{\nu}^{\star}_t\,\d t + \bar{Q}^{B,\star}_t\, \d g^c_t + g^c_t\,\left(\nu^{B,\star}_t - \bar{\nu}^{\star}_t\right)\d t + \sigma^\alpha g^a_t \d W^\alpha_t\\
&=  \alpha_t \d g^a_t - k^\alpha \alpha_t g^a_t \d t + \bar{Q}^{\star}_t \,\d g^b_t  + g^b_t \,\left(g^a_t \alpha_t + g^b_t\,\bar{Q}^{\star}_t + g^c_t\,\bar{Q}^{B,\star}_t\right)\,\d t + \bar{Q}^{B,\star}_t\, \d g^c_t\\
&\quad + g^c_t\,\left( \left(h^a_t - g^a_t\right) \alpha_t + \left(h^b_t - g^b_t\right)\,\bar{Q}^{\star}_t + \left(h^c_t - g^c_t\right)\,\bar{Q}^{B,\star}_t\right)\d t + \sigma^\alpha g^a_t \d W^\alpha_t\\
& = \alpha_t \big\{ \d g^a_t - k^\alpha  g^a_t \d t + g^b_t\,g^a_t\d t + g^c_t\left(h^a_t  - g^a_t\right)\d t \big\}+ \bar{Q}^{\star}_t\,\big\{ \d g^b_t + \left(g^b_t\right)^2\d t + g^c_t\left(h^b_t  - g^b_t\right)\d t\big\}\\
&\quad + \bar{Q}^{B,\star}_t\,\big\{ \d g^c_t + g^b_t\,g^c_t\d t + g^c_t\left(h^c_t  - g^c_t\right)\d t \big\}\,+ \sigma^\alpha g^a_t \d W^\alpha_t,
\end{align*}
and similarly, 
\begin{align*}
\d \bar{\nu}^{B,\star}_t &=  \alpha_t\d h^a_t + h^a_t \d \alpha_t + \bar{Q}^{\star}_t \,\d h^b_t  + h^b_t \,\d \bar{Q}^{\star}_t + \bar{Q}^{B,\star}_t\, \d h^c_t + h^c_t\,\d \bar{Q}^{B,\star}_t\\
&=  \alpha_t \d h^a_t - k^\alpha \alpha_t h^a_t \d t  + \bar{Q}^{\star}_t \,\d h^b_t  + h^b_t \,\bar{\nu}^{\star}_t\,\d t + \bar{Q}^{B,\star}_t\, \d h^c_t + h^c_t\,\left(\nu^{B,\star}_t - \bar{\nu}^{\star}_t\right)\d t + \sigma^\alpha h^a_t \d W^\alpha_t\\
&=  \alpha_t \d h^a_t - k^\alpha \alpha_t h^a_t \d t + \bar{Q}^{\star}_t \,\d h^b_t  + h^b_t \,\left(g^a_t \alpha_t + g^b_t\,\bar{Q}^{\star}_t + g^c_t\,\bar{Q}^{B,\star}_t\right)\,\d t + \bar{Q}^{B,\star}_t\, \d h^c_t\\
&\quad + h^c_t\,\left(\left(h^a_t - g^a_t\right)\alpha_t + \left(h^b_t - g^b_t\right)\,\bar{Q}^{\star}_t + \left(h^c_t - g^c_t\right)\,\bar{Q}^{B,\star}_t\right)\d t + \sigma^\alpha h^a_t \d W^\alpha_t\\
& = \alpha_t \big\{ \d h^a_t - k^\alpha h^a_t \d t + h^b_t\,g^a_t\d t + h^c_t\left(h^a_t  - g^a_t\right) \d t\big\} + \bar{Q}^{\star}_t\,\big\{ \d h^b_t + h^b_t\,g^b_t\d t + h^c_t\left(h^b_t  - g^b_t\right)\d t\big\}\\
&\quad + \bar{Q}^{B,\star}_t\,\big\{ \d h^c_t + h^b_t\,g^c_t\d t + h^c_t\left(h^c_t  - g^c_t\right)\d t \big\}\, + \sigma^\alpha h^a_t \d W^\alpha_t.
\end{align*}
Given that $\bar{\nu}^{\star}_t$ also satisfies the above FBSDE, we have that 
\begin{align*}
      \d\bar{\nu}^{\star}_t &= -\frac{1}{2\,\eta^I}\left(b\,\nu^{B,\star}_t  + \alpha_t - 2\,\bar \phi \bar{Q}^{\star}_t \right) \d t + \frac{1}{2\,\eta^I}\, \d \bar{Z}^I_t,\\
      &= -\frac{1}{2\,\eta^I}\left(b\,\left(h^a_t \alpha_t + h^b_t\,\bar{Q}^{\star}_t + h^c_t\,\bar{Q}^{B,\star}_t\right)  + \alpha_t - 2\,\bar \phi \bar{Q}^{\star}_t \right) \d t + \frac{1}{2\,\eta^I}\, \d \bar{Z}^I_t,
\end{align*}
and similarly for $\bar{\nu}^{B,\star}_t$, for which we have that 
\begin{align*}
      \d \nu^{B,\star}_t &= -\frac{1}{2\,\eta^B}\left(b\,\bar \nu^{\star}_t + \alpha_t  - 2 \phi^B \bar Q^{B,\star}_t \right) \d t + \frac{1}{2\,\eta^B} \d Z^B_t\\
      &= -\frac{1}{2\,\eta^B}\left(b\,\left(g^a_t \alpha_t + g^b_t\,\bar{Q}^{\star}_t + g^c_t\,\bar{Q}^{B,\star}_t\right) + \alpha_t  - 2 \phi^B \bar Q^{B,\star}_t \right) \d t + \frac{1}{2\,\eta^B} \d Z^B_t\,.
\end{align*}
Combining the derived expressions we have that 
\begin{align*}
    0 & = \alpha_t\bigg\{ \d g^a_t - k^\alpha g^a_t \d t + g^b_t\,g^a_t \d t + g^c_t\left(h^a_t  - g^a_t\right) \d t  + \frac{b\,h^a_t + 1}{2\,\eta^I} \d t\bigg\}\\
    &\quad + \bar{Q}^{\star}_t\,\bigg\{ \d g^b_t + \left(g^b_t\right)^2 \d t + g^c_t\left(h^b_t  - g^b_t\right) \d t + \frac{b\,h^b_t - 2\,\bar \phi}{2\,\eta^I} \d t \bigg\}\\
    &\quad + \bar{Q}^{B,\star}_t\,\bigg\{ \d g^c_t + g^b_t\,g^c_t \d t + g^c_t\left(h^c_t  - g^c_t\right) \d t + \frac{b\,h^c_t}{2\,\eta^I} \d t \bigg\}\, + \sigma^\alpha g^a_t \d W^\alpha_t- \frac{1}{2\,\eta^I}\,\d \bar{Z}^I_t,
\end{align*}
and 
\begin{align*}
    0 & = \alpha_t \bigg\{ \d h^a_t - k^\alpha h^a_t \d t + h^b_t\,g^a_t \d t + h^c_t\left(h^a_t  - g^a_t\right) \d t + \frac{b\,g^a_t+1}{2\,\eta^B} \d t  \bigg\} \\
&\quad + \bar{Q}^{\star}_t\,\bigg\{ \d h^b_t + h^b_t\,g^b_t \d t + h^c_t\left(h^b_t  - g^b_t\right) \d t + \frac{b\,g^b_t}{2\,\eta^B} \d t \bigg\}\\
&\quad + \bar{Q}^{B,\star}_t\,\bigg\{ \d h^c_t + h^b_t\,g^c_t \d t  + h^c_t\left(h^c_t  - g^c_t\right) \d t + \frac{b\,g^c_t - 2\,\phi^B}{2\,\eta^B}  \d t \bigg\}\, +\sigma^\alpha h^a_t \d W^\alpha_t- \frac{1}{2\,\eta^B}\d Z^B_t.
\end{align*}
Then, by setting
$$\d Z^B_t = 2 \eta^B \sigma^\alpha \bar h^a(t) \d W^\alpha_t\qquad \text{and} \qquad \d \bar Z^I_t = 2 \eta^I \sigma^\alpha \bar g^a(t) \d W^\alpha_t,$$
we observe that the system of equations becomes 
\begin{align*}
    0&=\d g^a_t + \left[- k^\alpha g^a_t +g^b_t\,g^a_t + g^c_t\left(h^a_t  - g^a_t\right)  + \frac{b\,h^a_t + 1}{2\,\eta^I}\right]\d t \\
    0&=\d h^a_t + \left[- k^\alpha h^a_t + h^b_t\,g^a_t + h^c_t\left(h^a_t  - g^a_t\right) + \frac{b\,g^a_t+1}{2\,\eta^B}\right] \d t \\
    0&=\d g^b_t + \left[\left(g^b_t\right)^2 + g^c_t\left(h^b_t  - g^b_t\right) + \frac{b\,h^b_t - 2\,\bar \phi}{2\,\eta^I} \right]\d t\\
    0&= \d h^b_t + \left[h^b_t\,g^b_t + h^c_t\left(h^b_t  - g^b_t\right) + \frac{b\,g^b_t}{2\,\eta^B}\right]\d t\\
    0&=\d g^c_t + \left[g^b_t\,g^c_t + g^c_t\left(h^c_t  - g^c_t\right) + \frac{b\,h^c_t}{2\,\eta^I}\right]\d t\\
    0&=\d h^c_t + \left[h^b_t\,g^c_t + h^c_t\left(h^c_t  - g^c_t\right) + \frac{b\,g^c_t - 2\,\phi^B}{2\,\eta^B} \right]\d t\,,
\end{align*}
 with terminal condition $g^a_T = h^a_T = g^c_T = h^b_T = 0$, $g^b_T = -\bar a /\eta^I$ and  $h^c_T = -(2\,a^B-b) /2\,\eta^B$. We see that the system for $g^b_t,g^c_t,h^b_t,h^c_t$ is independent of the solution to $g^a_t,h^a_t$.
Let $\bm{P}:[0,T]\to\R^4$ be given by
\begin{align*}
    \bm{P}_t = -\begin{pmatrix}
           h^c_t & h^b_t \\
           g^c_t & g^b_t 
         \end{pmatrix}
\end{align*}
and let $\bm{U},\bm{Y},\bm{Q},\bm{S}\in\R^{2\times 2}$ be given by
\begin{align*}
    \bm{U} = 
    \begin{pmatrix}
       1 & -1 \\
       0 &  1 
    \end{pmatrix}\,,\qquad
    \bm{Y} = 
    \begin{pmatrix}
       0 & \frac{b}{2\,\eta^B} \\
       \frac{b}{2\,\eta^I} &  0 
    \end{pmatrix}\,,\qquad
    \bm{Q} = 
    \begin{pmatrix}
       - \frac{\phi^B}{\eta^B} & 0 \\
       0 &  - \frac{\bar \phi}{\eta^I}
    \end{pmatrix}\,,\qquad 
    \bm{S} = \begin{pmatrix}
       \frac{2\,a^B-b}{2\,\eta^B} & 0 \\
       0 & \frac{\bar a}{\eta^I} 
    \end{pmatrix}\,.
\end{align*}
The system of ODEs for $g^b_t,g^c_t,h^b_t,h^c_t$ can be written as the following matrix Riccati differential  equation (MRDE)
\begin{align}\label{eq: riccati eq}
\begin{cases}
    & 0 = \frac{\d \bm{P}_t}{\d t} + \bm{Y}\,\bm{P}_t -  \bm{P}_t\,\bm{U}\,\bm{P}_t - \bm{Q}\,,\qquad t\in[0,T]\,,\\
    &\bm{P}_T = 
    \bm{S}\,.
\end{cases}  
\end{align}
The above system has a solution as a consequence of Theorem 2.3 in \citeauthor{freiling2000non} \cite{freiling2000non} for $\bm{C}=0$ and $\bm{D}=\bm{I}_2$, where $0$ and $\bm{I}_2$ denote the zero and the identity matrix in $\R^{2\times 2}$. To be more precise, using the notation of \cite{freiling2000non},  we have that  $\bm{B}_{11} = 0$, $\bm{B}_{12} = -\bm{U}$, $\bm{B}_{21}= \bm{Q}$, and $\bm{B}_{22} = -\bm{Y}$. Then, it follows that for our choice of $\bm{C}$ and $\bm{D}$, 
\begin{equation*}
    \bm{C}+\bm{D}\,\bm{S} + \bm{S}^\intercal\,\bm{D}^\intercal = 2\,\bm{S} >0\,,
\end{equation*}
that is, it is positive definite given that we assumed that  $2\,a^B-b > 0$ and $a^I>0$. Next, the matrix $\bm{L}$ defined as 
\begin{equation*}
    \bm{L}=  
    \begin{pmatrix}
      \bm{D}\,\bm{B}_{21} & \bm{B}_{11}^\intercal \,\bm{D} + \bm{D}\,\bm{B}_{22}\\
      0 & \bm{B}_{12}^{\intercal}\,\bm{D}
    \end{pmatrix} = 
    \begin{pmatrix}
      \bm{Q} & -\bm{Y}\\
      0 & -\bm{U}
    \end{pmatrix} 
    \,,
\end{equation*}
satisfies that $\det(\bm{L}) = \det(\bm{Q})\times\det(-\bm{U})$. Given that $\bm{Q}\leq 0$ and $-\bm{U}<0$, it follows that all eigenvalues of $\bm{L}$ are guaranteed to be non-positive, and at least one of them is guaranteed to be nonzero. A short calculation shows that $\bm{x}\,\left(\bm{L}+\bm{L}^\intercal\right)\,\bm{x}^\intercal$ is 
\begin{equation*}
    -\frac{\phi^B\,x_1^2}{\eta^B} -\frac{\phi^I\,x_2^2}{\eta^I} - \frac{b\,x_2\,x_3}{2\,\eta^I}  - x_3^2 - \frac{b\,x_1\,x_4}{2\,\eta^B} + x_3\,x_4 - x_4^2
\end{equation*}
where $\bm{x} = (x_1,x_2,x_3,x_4)$. Using the inequality $\pm 2\,x\,y \leq x^2+y^2$ for $x,y\in\mathbb{R}$, we see that a sufficient condition for $\bm{x}\,\left(\bm{L}+\bm{L}^\intercal\right)\,\bm{x}^\intercal\leq 0$ for all $\bm{x}\in\mathbb{R}^{4}$ is that $b\leq 2\,\eta^B,\,2\,\eta^I,\,4\,\phi^B,\,4\,\bar{\phi}$, which we assumed.
It follows that $\bm{L}+\bm{L}^\intercal\leq 0$ which implies that we can make use of Theorem 2.3 in \citeauthor{freiling2000non} \cite{freiling2000non} and show that there is a solution to \eqref{eq: riccati eq}.
These last arguments follows closely part II of the proof of Theorem 3.5 in \citeauthor{casgrain2020mean} \cite{casgrain2020mean}, and similar to them, given that the solution exists and is continuous in $[0,T]$, it is bounded, and  we conclude that the unique solution takes the form 
\begin{align*}
    \bm{P}_t = \bm{T}_t \,\bm{R}^{-1}_t\,,
\end{align*}
where $\bm{R}_t,\bm{T}_t$ solve the linear system of differential equations
\begin{align*}
    \frac{\d }{\d t}
    \begin{pmatrix}
        \bm{R}_t\\
        \bm{T}_t
    \end{pmatrix}
    = 
    \begin{pmatrix}
        0 & \bm{U}\\
        -\bm{Q} & -\bm{Y}
    \end{pmatrix}
    \,
    \begin{pmatrix}
        \bm{R}_t\\
        \bm{T}_t
    \end{pmatrix}\,,\qquad \begin{pmatrix}
        \bm{R}_T\\
        \bm{T}_T
    \end{pmatrix} = \begin{pmatrix}
        I\\
        \bm{S}
    \end{pmatrix}\,,
\end{align*}
as a consequence of Theorem 3.1 in \citeauthor{freiling2002survey} \cite{freiling2002survey}.\\

Lastly, given the above solutions, we just have to solve the linear system of ODEs given by:
\begin{align}\label{haga}
\begin{cases}
     0&=\d g^a_t + \left[ - k^\alpha g^a_t + g^b_t\,g^a_t + g^c_t\left(h^a_t  - g^a_t\right)  + \frac{b\,h^a_t + 1}{2\,\eta^I}\right] \d t \vspace{0.1cm} \\
     0&=\d h^a_t +  \left[ - k^\alpha h^a_t + h^b_t\,g^a_t + h^c_t\left(h^a_t  - g^a_t\right) + \frac{b\,g^a_t+1}{2\,\eta^B} \right] \d t   \,,
\end{cases}
\end{align}
with terminal conditions $g^a_T = h^a_T =0$. Let 
\begin{align*}
\bm{X}_t = \begin{pmatrix}
				h^a_t\\
				g^a_t
			\end{pmatrix}\,,
\qquad \bm{A}_t = \begin{pmatrix}
				-\frac{1}{2\,\eta^B} \vspace{0.2cm} \\
				-\frac{1}{2\,\eta^I}
			\end{pmatrix}\,,
\qquad \bm{B}_t = \begin{pmatrix}
				k^\alpha -h^c_t & h^c_t - h^b_t - \frac{b}{2\,\eta^B} \\
				-g^c_t - \frac{b}{2\,\eta^I} & k^\alpha + g^c_t - g^b_t
			\end{pmatrix}\,,
\end{align*}
then, we have that the system for $h^a_t$ and $g^a_t$ can be written as
\begin{align*}
\d \bm{X}_t = \left(\bm{A}_t +  \bm{B}_t\,\bm{X}_t\right)\d t\,,
\end{align*}
with terminal condition $\bm{X}_T = 0$. The solution is therefore given by
\begin{align*}
    \bm{X}_t =  - \int_t^T e^{-\int_t^s \bm B_u \d u}\bm A_s \d s.
\end{align*}

The closed-form optimal solution to \eqref{MFGFBSDE_system2} is then 
\begin{align}\label{eq: closed-form MFG strat}
\begin{pmatrix}
	\nu^{B,\star}_t\\
	\bar{\nu}^{\star}_t 
\end{pmatrix} = \bm{X}_t\, \alpha_t - \bm{P}_t\,\begin{pmatrix}
	\bar{Q}^{B,\star}_t\\
	\bar{Q}^{\star}_t 
\end{pmatrix}\,.
\end{align}

\begin{remark}
The optimal trading strategy of the broker can be written as
\begin{align}
    \nu^{B,\star}_t &=  q^a_t\,\left(\bar{\nu}^{\star}_t - g^b_t\,\bar{Q}^{\star}_t - g^c_t\,\bar{Q}^{B,\star}_t\right) + h^b_t\,\bar{Q}^{\star}_t + h^c_t\,\bar{Q}^{B,\star}_t\\
    &=  q^a_t\,\bar{\nu}^{\star}_t   + \left(h^b_t- q^a_t\,g^b_t\,\right)\,\bar{Q}^{\star}_t + \left(h^c_t- q^a_t\,g^c_t\right)\,\bar{Q}^{B,\star}_t\,,
\end{align}
where the externalisation rate $q^a_t$ is defined as
\begin{equation}
    q^a_t = \frac{h^a_t}{g^a_t}
\end{equation}
for $t\in[0,T)$, and for $t=T$ as the limit of the above expression when $t\to T$.
\end{remark}

\subsubsection{The optimal strategy of the informed trader}

Finally, we can solve the FBSDE of the representative informed trader:
\begin{align}\label{FBSDE_informed}
    \begin{cases}
    -\d \left(2\,\eta^I \nu^{I,\star}_t \right) &= \left(b\,\nu^{B,\star}_t + \alpha^I_t + \alpha_t - 2\,\phi^I Q^{I,\star}_t \right) \d t - \d Z^I_t,\\
        2\,\eta^I \nu^{I,\star}_T &= -2\,a^I Q^{I, \star}_T.
\end{cases}
\end{align}
As before, we make an ansatz and look for a solution with the form
\begin{align}\label{eq: ind informed opt strat}
    \nu^{I,\star}_t = f^a_t \alpha_t + f^{a,I}_t \alpha^I_t + f^b_t \bar Q^\star_t + f^{b,I}_t Q^{I,\star}_t + f^c_t \bar Q^{B,\star}_t,
\end{align}
where $f^a, f^{a,I}, f^b, f^{b,I}, f^c$ are deterministic $\mathcal C^1$ functions, with terminal conditions $f^a_T= f^{a,I}_T= f^b_T=  f^c_T= 0$ and $f^{b,I}_T = -a^I/\eta^I$, and where
$$Q^{I,\star}_t = \int_0^t \nu^{I,\star}_u \d u.$$

It then follows that
\begin{align*}
    \d \nu^{I,\star}_t &= \alpha_t \d f^a_t + f^a_t \d \alpha_t + \alpha^I_t \d f^{a,I}_t + f^{a,I}_t \d \alpha^I_t + \bar Q^\star_t \d f^b_t + f^b_t \d \bar Q^\star_t + Q^{I,\star}_t \d f^{b,I}_t + f^{b,I}_t \d Q^{I,\star}_t + \bar Q^{B,\star}_t \d f^c_t + f^c_t \d \bar Q^{B,\star}_t\\
    &= \alpha_t \d f^a_t - k^{\alpha} \alpha_t f^a_t \d t + \alpha^I_t \d f^{a,I}_t - k^I \alpha^I_t f^{a,I}_t \d t + \bar Q^\star_t \d f^b_t + f^b_t \bar \nu^\star_t \d t + Q^{I,\star}_t \d f^{b,I}_t + f^{b,I}_t \nu^{I,\star}_t \d t + \bar Q^{B,\star}_t \d f^c_t\\
    & \quad + f^c_t \left( \nu^{B,\star}_t -  \nu^\star_t \right) \d t + \sigma^\alpha f^a_t \d W^\alpha_t + \sigma^I f^{a,I}_t \d W^I_t\\
    &= \alpha_t \d f^a_t - k^{\alpha} \alpha_t f^a_t \d t + \alpha^I_t \d f^{a,I}_t - k^I \alpha^I_t f^{a,I}_t \d t + \bar Q^\star_t \d f^b_t + f^b_t \left( g^a_t \alpha_t + g^b_t \bar Q^\star_t + g^c_t \bar Q^{B,\star}_t \right) \d t\\
    &\quad + Q^{I,\star}_t \d f^{b,I}_t + f^{b,I}_t \left( f^a_t \alpha_t + f^{a,I}_t \alpha^I_t + f^b_t \bar Q^\star_t + f^{b,I}_t Q^{I,\star}_t + f^c_t \bar Q^{B,\star}_t \right) \d t + \bar Q^{B,\star}_t \d f^c_t\\
    &\quad + f^c_t \left( (h^a_t - g^a_t) \alpha_t + (h^b_t-g^b_t) \bar Q^\star_t + (h^c_t - g^c_t) \bar Q^{B,\star}_t \right) \d t + \sigma^\alpha f^a_t \d W^\alpha_t + \sigma^I f^{a,I}_t \d W^I_t\\
    &= \alpha_t \left\{\d f^a_t - k^\alpha f^a_t \d t +f^b_t g^a_t \d t + f^{b,I}_t f^a_t \d t + f^c_t (h^a_t - g^a_t) \d t \right\}\\
    &\quad + \alpha^I_t \left\{ \d f^{a,I}_t -k^I f^{a,I}_t \d t + f^{b,I}_t f^{a,I}_t \d t \right\}\\
    &\quad + \bar Q^\star_t \left\{ \d f^b_t + f^b_t g^b_t \d t + f^{b,I}_t f^b_t \d t + f^c_t (h^b_t - g^b_t) \d t\right\}\\
    &\quad + Q^{I,\star}_t \left\{ \d f^{b,I}_t + \left(f^{b,I}_t\right)^2 \right\} + \bar Q^{B,\star}_t \left\{ \d f^c_t + f^b_t g^c_t \d t + f^{b,I}_t f^c_t \d t + f^c_t (h^c_t - g^c_t) \ d t\right\}\\
    &\quad + \sigma^\alpha f^a_t \d W^\alpha_t + \sigma^I f^{a,I}_t \d W^I_t.
\end{align*}

Given that $\nu^{I,\star}_t$ also satisfies the above FBSDE, we have that
\begin{align*}
    \d \nu^{I,\star}_t &= -\frac{1}{2\,\eta^I}\left(b\,\nu^{B,\star}_t  + \alpha^I_t + \alpha_t - 2\,\phi^I 
    {Q}^{I, \star}_t \right) \d t + \frac{1}{2\,\eta^I}\, \d {Z}^I_t,\\
      &= -\frac{1}{2\,\eta^I}\left(b\,\left(h^a_t \alpha_t + h^b_t\,\bar{Q}^{\star}_t + h^c_t\,\bar{Q}^{B,\star}_t\right) + \alpha^I_t  + \alpha_t - 2\,\phi^I {Q}^{I,\star}_t \right) \d t + \frac{1}{2\,\eta^I}\, \d {Z}^I_t.
\end{align*}

Combining the derived expressions we have that
\begin{align*}
    0 &= \alpha_t \left\{ \d f^a_t - k^\alpha f^a_t \d t +f^b_t g^a_t \d t + f^{b,I}_t f^a_t \d t + f^c_t (h^a_t-g^a_t) \d t + \frac{b h^a_t + 1}{2\eta^I} \d t \right\}\\
    &\quad + \alpha^I_t \left\{ \d f^{a,I}_t -k^I f^{a,I}_t \d t + f^{b,I}_t f^{a,I}_t \d t +\frac 1{2\eta^I} \d t \right\}\\
    &\quad + \bar Q^\star_t \left\{ \d f^b_t + f^b_t g^b_t \d t +f^{b,I}_t f^b_t \d t + f^c_t (h^b_t - g^b_t) \d t + \frac{b h^b_t}{2 \eta^I} \d t \right\}\\
    &\quad + Q^{I,\star}_t \left\{ \d f^{b,I}_t + \left( f^{b,I}_t \right)^2 \d t - \frac{\phi^I}{\eta^I} \d t \right\}\\
    &\quad + \bar Q^{B,\star}_t \left\{ \d f^c_t +f^b_t g^c_t \d t +f^{b,I}_t f^c_t \d t + f^c_t \left( h^c_t - g^c_t \right) \d t + \frac{bh^c_t}{2 \eta^I} \d t \right\}\\
    &\quad + \left[ \sigma^\alpha f^a_t \d W^\alpha_t + \sigma^I f^{a,I}_t \d W^I_t - \frac 1{2\eta^I} \d Z^I_t \right].
\end{align*}

Then, by setting
$$\d Z^I_t  = 2\eta^I \left[ \sigma^\alpha f^a_t \d W^\alpha_t + \sigma^I f^{a,I}_t \d W^I_t \right],$$
we observe that the system of equations becomes
\begin{align}
    0&=\d f^a_t + \left[- k^\alpha f^a_t +f^b_t g^a_t + f^{b,I}_t f^a_t + f^c_t (h^a_t-g^a_t) + \frac{b h^a_t + 1}{2\eta^I} \right]\d t \label{fa}\\
    0&=\d f^{a,I}_t + \left[ -k^I f^{a,I}_t + f^{b,I}_t f^{a,I}_t +\frac 1{2\eta^I} \right] \d t \label{faI}\\
    0&=\d f^b_t + \left[f^b_t g^b_t +f^{b,I}_t f^b_t + f^c_t (h^b_t - g^b_t) + \frac{b h^b_t}{2 \eta^I} \right]\d t\label{fb}\\
    0&= \d f^{b,I}_t + \left[\left( f^{b,I}_t \right)^2 - \frac{\phi^I}{\eta^I}\right]\d t\label{fbI}\\
    0&=\d f^c_t + \left[f^b_t g^c_t +f^{b,I}_t f^c_t + f^c_t \left( h^c_t - g^c_t \right) + \frac{bh^c_t}{2 \eta^I}\right]\d t,\label{fc}
\end{align}
with terminal conditions $f^a_T= f^{a,I}_T= f^b_T=  f^c_T= 0$ and $f^{b,I}_T = -a^I/\eta^I$.\\

Notice that Equation \eqref{fbI} is independent from the rest of the system. A particular solution to this Riccati equation is given by
$$y_p(t) = -\sqrt{\frac{\phi^I}{\eta^I}} \tanh \left( \sqrt{\frac{\phi^I}{\eta^I}} (T-t) \right) \qquad \forall t \in [0,T].$$

We then know that the general solution is given by
$y_p + u$ where $u$ solves
$$u' = -u^2 - 2y_p u$$
on $[0,T]$. Substituting $z = 1/u$ yields
$$z' = 2y_pz + 1,$$
with the terminal condition $f^{b,I}_T = - a^I / \eta^I $ now translating as $z(T) = -\eta^I / a^I$. The solution to this linear ODE is given by 
$$z(t) = -\frac{\eta^I}{a^I} \exp \left( -2 \int_t^T y_p(s) \d s \right) - \int_t^T \exp \left( -2 \int_t^u y_p(s) \d s \right) \d u .$$

We can finally conclude that the unique solution to Equation \eqref{fbI} with terminal condition $f^{b,I}_T = - a^I / \eta^I $
is given by
\begin{align*}
    f^{b,I}_t = -\sqrt{\frac{\phi^I}{\eta^I}} \tanh \left( \sqrt{\frac{\phi^I}{\eta^I}} (T-t) \right) - \frac{e^{2\int_t^T y_p(s) \d s}}{\eta^I / a^I + \int_t^T e^{2\int_u^T y_p(s) \d s} \d u }
\end{align*}
for $t\in [0,T]$.\\

Once we know $f^{b,I}$, Equation \eqref{faI} is a simple linear ODE with terminal condition $f^{a,I}_T = 0$. Its solution for $t\in [0,T]$ is therefore given by
\begin{align*}
    f^{a,I}_t = \frac 1{2\eta^I} \int_t^T e^{-\int_t^u \left( k^I - f^{b,I}_s \right) \d s} \d u\,.
\end{align*}

Let $\bm{A}^{b,c} :  [0,T]\rightarrow \mathbb R^4$ and $\bm{b}^{b,c} :  [0,T]\rightarrow \mathbb R^2$ be given by
$$\bm{A}^{b,c}_t = - \begin{pmatrix}  g^b_t + f^{b,I}_t  &  h^b_t - g^b_t  \\
g^c_t &      h^c_t - g^c_t + f^{b,I}_t
\end{pmatrix} \qquad \text{and} \qquad \bm{b}^{b,c}_t = - \frac{b}{2\eta^I} \begin{pmatrix}
    h^b_t \\ h^c_t
\end{pmatrix}.$$
We introduce the function $\bm{F}^{b,c} :  [0,T]\rightarrow \mathbb R^2$ given by
$$\bm{F}^{b,c}_t = \begin{pmatrix}
    f^b_t \\ f^c_t
\end{pmatrix}.$$

Then it is clear from Equations \eqref{fb} and \eqref{fc} that $\bm{F}^{b,c}$ satisfies
\begin{align*}
    \frac{\d}{\d t}{\bm{F}^{b,c}_t} = \bm{A}^{b,c}_t \bm{F}^{b,c}_t + \bm{b}^{b,c}_t
\end{align*}
with terminal condition $\bm{F}^{b,c}_T = 0$. The solution is 
\begin{align*}
    \bm{F}^{b,c}_t = - \int_t^Te^{-\int_t^u \bm{A}^{b,c}_s \d s } \bm{b}^{b,c}_u \d u
\end{align*}
for $t\in [0,T]$.\\

Finally, if we define $b^a:[0,T]\rightarrow \mathbb R$ by
$$b^a_t = -f^b_tg^a_t - f^c_t (h^a_t - g^a_t) - \frac{bh^a_t + 1}{2\eta^I} \qquad \forall t \in [0,T],$$

then the unique solution to the linear Equation \eqref{fa} with terminal condition $f^a_T = 0$ is given by
\begin{align*}
    f^a_t = -\int_t^T b^a_u e^{-\int_t^u (k^\alpha - f^{b,I}_s \d s)} \d u
\end{align*}
for  $t\in [0,T]$.

\section{Numerical results}\label{sec_num}
In this section we study the optimal trading strategies derived above. We discretise the trading window $[0,T]$, with $T=1$,   in  10,000 steps and perform one million simulations. Model parameters for the price dynamics are $\alpha_0 = 0$, $S_0 = 100$, $k^{\alpha} = 5$, $\sigma^{\alpha} = 1$, $\sigma^{s} = 1$. The price impact and penalty parameters are $\eta^I  = 1.0\times 10^{-3}$, $\eta^B = 1.2\times 10^{-3}$, $b=10^{-3}$, $a^I = 1$, $a^B = 1$, and $\phi^B=\phi^I= 10^{-2}$. 
Figure \ref{fig: paths} shows two sample paths of the main processes involved in the MFG Nash equilibrium we obtained in \eqref{eq: closed-form MFG strat}. 

\begin{figure}[H]
\begin{center}
\includegraphics[width=0.7\textwidth]{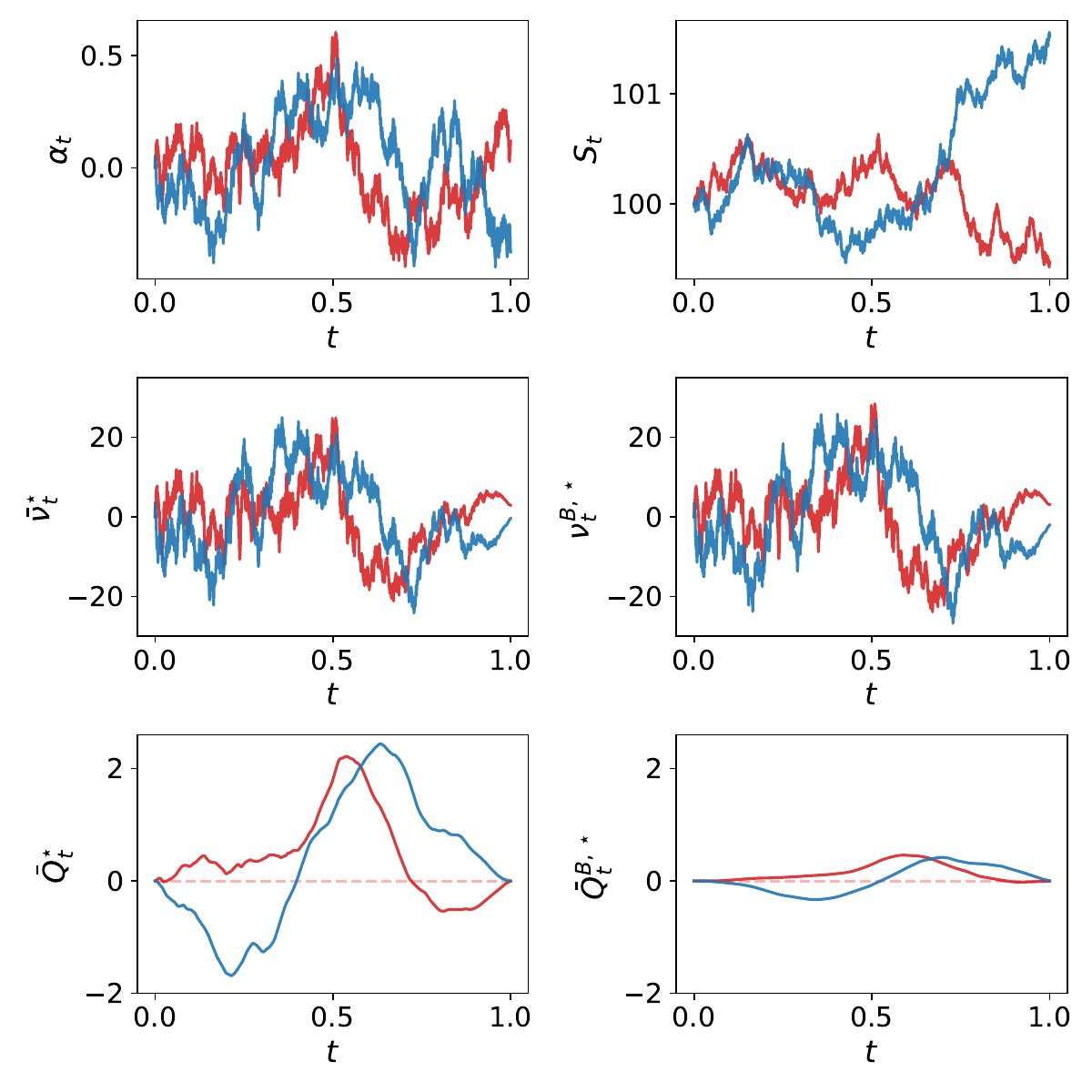}
\caption{Sample paths for $S_t$ $\alpha_t$, $\nu^I_t$, $\nu^B_t$, $Q^I_t$, and $Q^B_t$.}
\label{fig: paths}
\end{center}
\end{figure}

We see that both the mean field trading speed of the informed traders and that of the broker look almost identical to the naked eye for each of the two simulations shown. The cumulative difference, which of course is not zero, is shown in the inventory of the broker in the bottom left panel. \\

Next, we study each of the functions $g^{a,b,c}, h^{a,b,c}:[0,T]\to\mathbb{R}$ which define the optimal trading speeds in terms of the state variables of the control problems. Figure \ref{fig: functions g and h} shows each of the functions for the end of the trading window, in particular, we show the range $[0.95,1]$.

\begin{figure}[H]
\begin{center}
\includegraphics[width=0.7\textwidth]{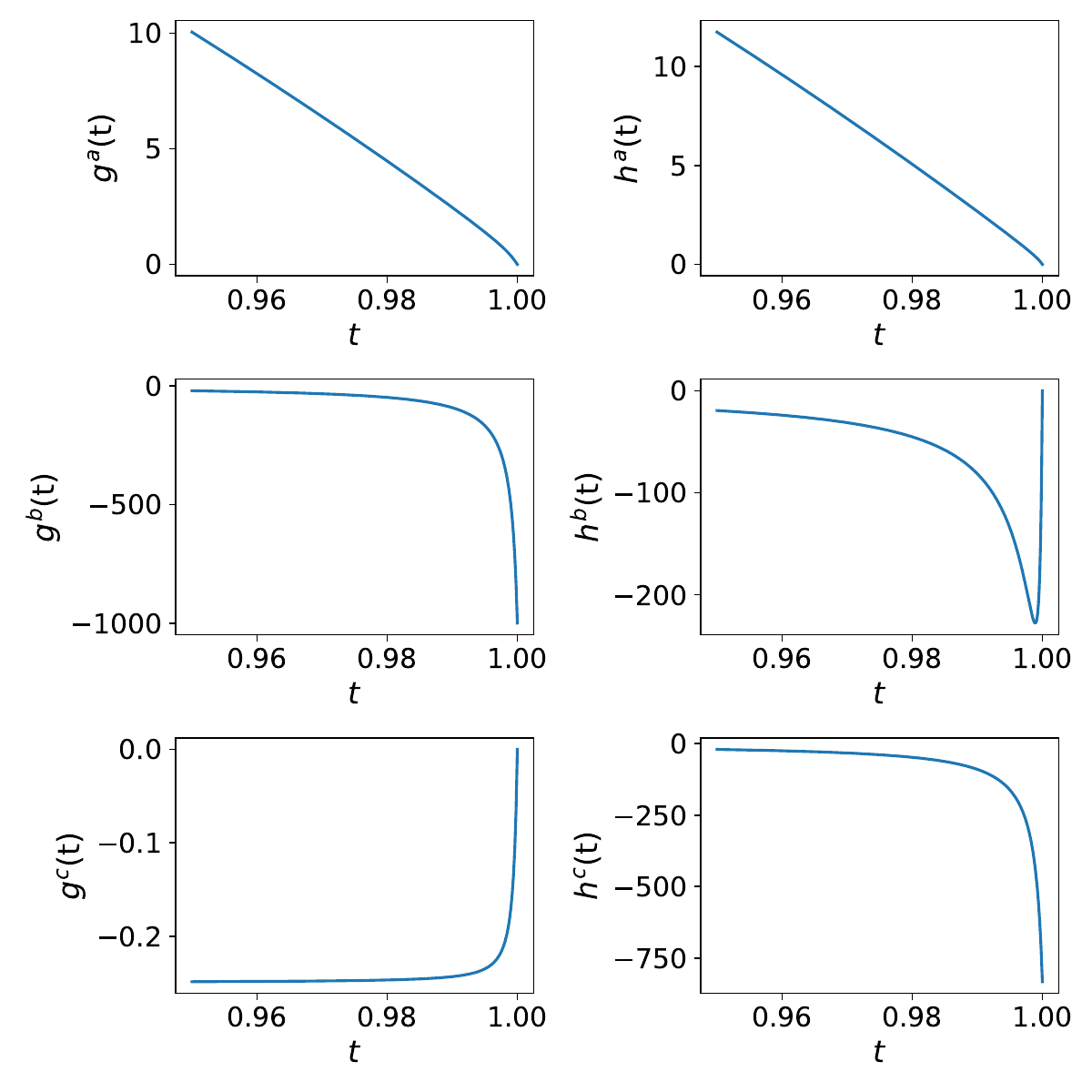}
\caption{Functions $g^{a,b,c}, h^{a,b,c}:[0,T]\to\mathbb{R}$ as time approaches $T$.}
\label{fig: functions g and h}
\end{center}
\end{figure}

We see that all four $g^{b,c},h^{b,c}$ are negative which follows from the intuition that the players wish to keep their inventory close to zero. Both $g^a$ and $h^a$ are positive and decrease towards zero which prescribes the way in which the signal is used, that is, both the informed and the broker trade in the direction of the common signal and as time progresses, this component of the trading strategy vanishes.
Both $g^b$ and $h^c$ have a similar behaviour; this is because these functions are the ones that force the terminal inventory (of the informed trader and the broker) towards the optimal level which gets closer to zero the larger the terminal penalty. Recall that $g^b(t)\,\bar{Q}^\star_t$ is part of the optimal trading speed of the mean-field informed trader and $h^c(t)\,\bar{Q}^{B,\star}_t$ is part of the optimal trading speed of the broker.\\ 

A more interesting behaviour is that of $h^b$. As expected, $h^b$ is negative. We observe that it decreases fast just before time $T.$ This is because of the terminal penalty of the informed traders; assume for instance that, as $t$ gets close to $T$, $\bar Q^\star_t$ is positive; in that case, the broker knows that, on average, the traders will start selling fast to him, because they want to have a flat inventory at $T$. Thus, in anticipation of this, the broker starts selling fast too on the D2D market. Lastly, the terminal condition takes $h^b$ back to zero. \\

Finally, Figure \ref{fig: functions fI} shows the components $f^{a,b,c}, f^{a,I}, f^{b,I}:[0,T]\to\mathbb{R}$ of the trading strategy of the individual informed trader; see \eqref{eq: ind informed opt strat}. We employ the same model parameters as before, together with $k^I = k^\alpha$, and $\sigma^I = 0.5\,\sigma^\alpha$. That is, the private signal has the same mean-reverting rate but lower variance when compared to the common signal.

\begin{figure}[H]
\begin{center}
\includegraphics[width=0.95\textwidth]{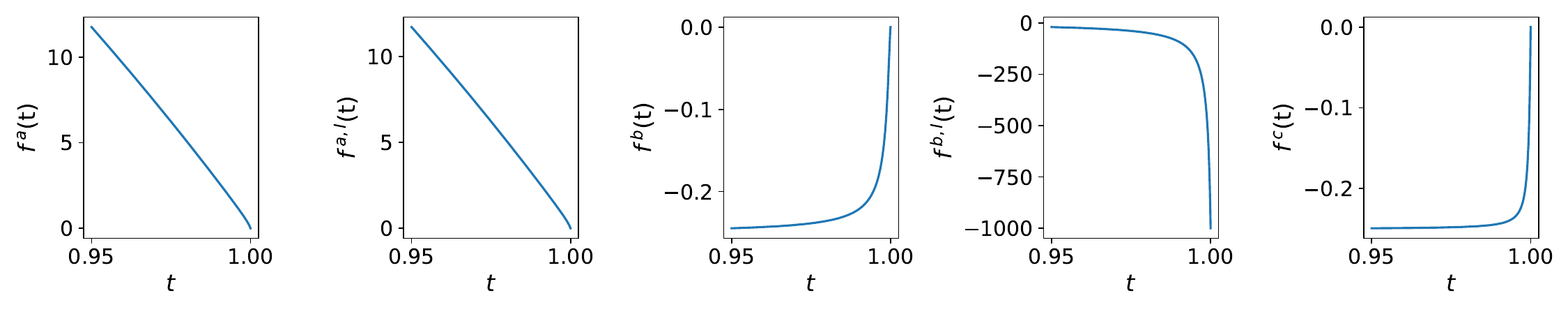}
\caption{Functions $f^{a,b,c}, f^{a,I}, f^{b,I}:[0,T]\to\mathbb{R}$ as time approaches $T$.}
\label{fig: functions fI}
\end{center}
\end{figure}

The interesting comparison is between (i) $f^a$ and $f^{a,I}$, and (ii) $f^b$ and $f^{b,I}$. We see that for (i) the behaviour is roughly the same. That is, the individual informed trader follows both signals in a similar way. On the other hand, the comparison for (ii) is not as straightforward. Indeed, as time progresses $f^{b,I}$ becomes more and more important in the trading strategy of the individual informed trader because of the constraint to liquidate inventory, whereas the value of $f^b$ vanishes because the informed trader stops pre-empting what the broker offloads of their order flow. 

\begin{figure}[H]
\begin{center}
\includegraphics[width=0.95\textwidth]{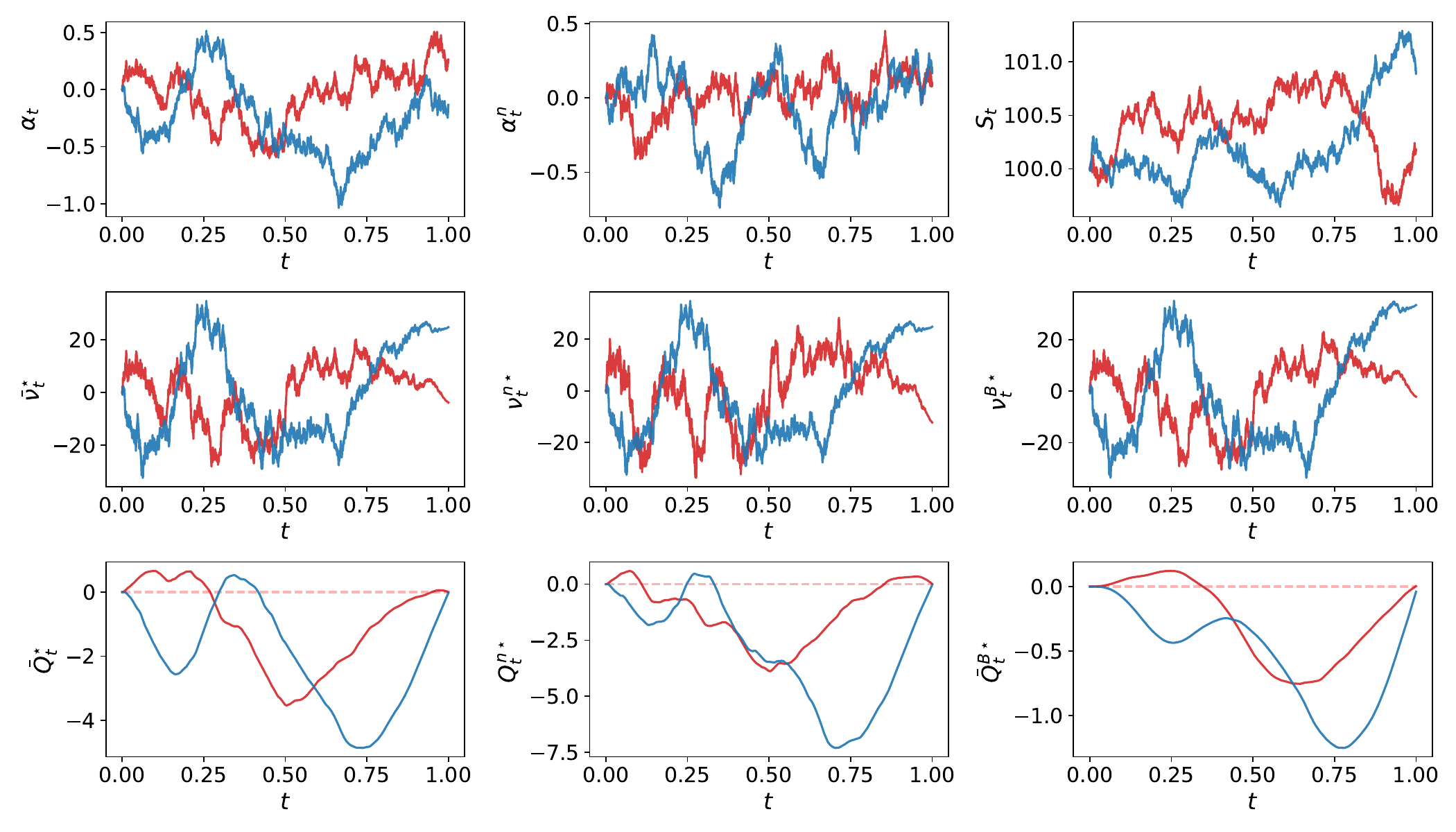}
\caption{Sample paths for the common signal $\alpha_t$ and the private signal $\alpha^n_t$, together with $S_t$, $\bar{\nu}^*_t$, $\nu^{n*}_t$, $\nu^{B,*}_t$, $\bar{Q}^*_t$, $Q^{n*}_t$, and $Q^{B,*}_t$.}
\label{fig: ith trader sample path}
\end{center}
\end{figure}

Figure \ref{fig: ith trader sample path} shows the effect of the individual signal $\alpha^n$ for the $n$-th informed trader. From the left middle panel and centre panel we observe that, both $\bar{\nu}^*_t$ and $\nu^{n*}_t$ are similar, with the latter showing a rougher behaviour due to the actions of the $n$-th informed trader on the  individual signal. The trajectory in red in the bottom two left panels shows the difference in more detail.

\section{Conclusion}\label{sec_concl}

In this paper, we study the problem of a broker facing many informed traders. Each informed trader observes both a common and an idiosyncratic signal. The broker charges a fixed transaction cost and chooses his externalisation rate based on the common signal and the average behaviour of the traders. Using a Gâteaux derivative approach, we derive a system of coupled forward-backward SDEs driving this optimisation problem. Using a sequence of  ansatzes, we solve this FBSDE system in closed form, and obtain the equilibrium strategy of the broker and that of the representative informed trader.\\

We then illustrate the results of our model using a set of realistic market parameters. As expected, the average trader's inventory moves with the common signal, and the broker adjusts his externalisation rate accordingly. More interestingly, the individual signal of a trader seems to have little impact on his trading strategy. This is due to the market impact of the broker as he externalises: his externalisation rate is driven by the average trading rate of the traders, which is itself driven by the common signal. Therefore, even in the presence of a private information that contradicts the beliefs of the market, the representative trader still tends to follow the herd, at least for a large enough value of the permanent price impact parameter.

\bibliographystyle{plainnat}
\bibliography{mfgbroker.bib}

\end{document}